\documentclass[10pt]{article}
\usepackage{amsmath,amsthm}
\usepackage{amsfonts}
\usepackage{bbm}
\usepackage{hyperref}

\usepackage{amssymb,epsfig}
\newtheorem{theorem}{Theorem}[section]
\newtheorem{defn}[theorem]{Definition}
\newtheorem{lemma}[theorem]{Lemma}
\newtheorem{corollary}[theorem]{Corollary}
\newtheorem{asmp}[theorem]{Assumption}
\newtheorem{example}[theorem]{Example}
\newtheorem{prob}[theorem]{Problem}
\newtheorem{remark}[theorem]{Remark}
\newtheorem{proposition}[theorem]{Proposition}
\usepackage{color}

\newcommand{\re}{\mathbb{R}}
\newcommand{\R}{\mathcal{R}}
\newcommand{\A}{\mathcal{A}}
\newcommand{\p}{\mathbb{P}}

\newcommand{\Q}{\mathbb{Q}}
\newcommand{\e}{\mathbb{E}}

\newcommand{\F}{\mathcal{F}}

\newcommand{\pspace}{(\Omega, \F,(\F_t),\p)}

\newcommand{\powut}[1]{H^{({#1})}}

\newcommand{\td}{\mathrm{d}}
\newcommand{\cer}{\mathrm{CER}}
\usepackage{textcomp}
\newcommand{\cerd}{\textnormal{\textcolonmonetary}\mathrm{ER}}
\newcommand{\tr}{'}

\DeclareMathOperator{\van}{\xrightarrow[n\to\infty]{}}

\DeclareMathOperator{\vax}{\xrightarrow[x\to\infty]{}}
\newcommand{\ovl}{\overline}

\title{Portfolio optimisation under non-linear drawdown constraints in a semimartingale financial model}
\author{Vladimir Cherny\thanks{e-mail: \texttt{Vladimir.Cherny@maths.ox.ac.uk}}\ \  and Jan Ob\l\'oj\thanks{e-mail:
        \texttt{obloj@maths.ox.ac.uk}; web:
        \texttt{www.maths.ox.ac.uk/$\sim$obloj/}}\smallskip\\
  Mathematical Institute \emph{and}\\ Oxford--Man Institute of Quantitative Finance\\
  University of Oxford,
  Oxford OX1 3LB, UK
}

\date{22 April 2013}

\begin{document}

\maketitle

\begin{abstract} 
A drawdown constraint forces the current wealth to remain above a given function of its maximum to date. We consider the portfolio optimisation problem of maximising the long-term growth rate of the expected utility of wealth subject to a drawdown constraint, as in the original setup of Grossman and Zhou (1993). We work in an abstract semimartingale financial market model with a general class of utility functions and drawdown constraints.
We solve the problem by showing that it is in fact equivalent to an unconstrained problem with a suitably modified utility function. Both the value function and the optimal investment policy for the drawdown problem are given explicitly in terms of their counterparts in the unconstrained problem.
\end{abstract}

\section{Introduction}
We study portfolio optimisation subject to drawdown constraints. A drawdown constraint specifies that the investor's wealth $V_t$ has to remain above a given function $w$ of its maximum to date: $V_t> w(\sup_{u\leq t} V_u)$. The motivating example is the case of a linear $w$, when the current wealth is always greater than a fixed fraction of its past maximum. Such features are often embedded in investment opportunities available in the financial markets. From the investor's perspective, they offer a partial protection of the realised gains, where the past maximum is viewed as a natural reference point. For a manager who is trading clients' money avoiding large drawdowns is crucial -- typically many investors have a \emph{stop-loss} provision and a large drawdown would result in a sudden withdrawal of capital from the fund, see Chekhlov et al.\ \cite{CUZ:05}.

This problem was originally introduced by Grossman and Zhou \cite{GZ} who considered a power utility investor in a Black-Scholes market who faces a linear drawdown constraint and maximises the long-term (asymptotic) growth rate of the expected utility of her wealth. Grossman and Zhou \cite{GZ} applied the \emph{forward approach} and solved the problem using the dynamic programming principle. Later Cvitani\'c and Karatzas \cite{CK} generalised the setting in \cite{GZ} to a complete $n$-dimensional market with deterministic coefficients. Using martingale theory they were able to link the solution to the optimisation problem with the drawdown constraint to an unconstrained problem which they could solve using the \emph{dual approach} as in Karatzas, Lehoczky and Shreve \cite{KaratzasLehochkyShreve:87}. The initial motivation for our work was to see if a similar link between the constrained and  unconstrained problems could be established in a much greater generality.

Inline with \cite{GZ,CK} we consider maximisation of the asymptotic growth rate of the expected utility of wealth. The idea to look at the long-run optimality has proven to be a powerful tool in solving various portfolio optimisation problems explicitly. It was used by Grossman and Vila \cite{GrossmanVila:92} in a problem with leverage constraints and by Dumas and Luciano \cite{DumasLuciano:91} in presence of transaction costs. It is close to the objectives studied in the \emph{risk-sensitive control}, see Section \ref{sec:examples}. We refer to Guasoni and Robertson \cite{GuasoniRobertson:11} for a more detailed discussion.

In this paper, we effectively solve the long-run continuous time portfolio optimisation problem with drawdown constraints.
More precisely, the main contribution of the paper is an equivalence result: the $w$-drawdown constrained problem with utility $U$ has the same value function as an unconstrained portfolio optimisation problem but with utility $U\circ F_w$, where $F_w$ is given explicitly in terms of $w$. Moreover, the optimal wealth process for the drawdown constrained problem is obtained as an explicit pathwise transformation $M^{F_w}(V^*)$ of the optimal wealth process $V^*$ for the unconstrained problem.
Both the function $F_w$ and the transform $M^{F_w}$ are given in terms of the drawdown constraint $w$ and are independent of the underlying semimartingale model. In consequence, endowing an investor with a drawdown constraint is an effective and model-independent way of encoding her preferences. It is equivalent to modifying her utility function, e.g. a (long-run) power utility investor with risk aversion $\rho\geq 0$, $U(x)=\frac{1}{1-\rho}x^{1-\rho}$, and a linear drawdown $w(x)=\alpha x$, $\alpha\in (0,1)$, has the same value function as an unconstrained investor with risk aversion $\alpha + \rho (1-\alpha)$. Note however that the two investors employ different optimal strategies. The former uses $M^{F_w}(V^*_\rho)$ which satisfies the $\alpha$-drawdown constraint, where $V^*_\rho$ would be her optimal strategy without drawdown constraints. The latter uses $V^*_{\alpha + \rho(1-\alpha)}$, which may have arbitrarily large drawdowns.

Our results hold in an abstract semimartingale model and the investor is endowed with a generic utility function $U$ and a drawdown constraint $w$. Specifically, we only assume that wealth processes are max-continuous (i.e.\ have a continuous running supremum), that $U$ either behaves like a logarithm or dominates a power function, and has a finite asymptotic elasticity as in Kramkov and Schachermayer \cite{KrSch}, and that $w(x)/x\in (0,1)$ is bounded away from $1$. In such a general setting there is little hope to solve portfolio optimisation problems explicitly.
Adding a drawdown constraint, which is a path-dependent constraint on the admissible investment strategies, appears to significantly increase the complexity. Rather surprisingly, our results show that this is not the case: the constrained problem is just as easy, or just as hard, as the analogue portfolio optimisation problem with no constraints. Since we consider the long-run optimality, Guasoni and Robertson \cite{GuasoniRobertson:11} show that the latter can be solved in a rather general diffusion setting, see Section \ref{sec:example_incomplete}.

This paper relies in an essential way on the so-called Az\'ema-Yor processes. They effectively provide us with a bijection between non-negative wealth processes and the wealth processes which satisfy a given drawdown constraint. Az\'ema-Yor martingales have initially appeared in \cite{AY} where they were used to solve the Skorokhod embedding problem. Carraro, El Karoui and Ob\l\'oj \cite{CEO} introduced a more general class of Az\'ema-Yor processes and studied them from an SDE perspective. In particular they investigated their properties in relation to drawdown constraints. These results provided crucial insights for our work. In fact, objects in Cvitani\'c and Karatzas \cite{CK} can be expressed using Az\'ema-Yor processes simplifying greatly the proofs in \cite{CK}, see Section \ref{sec:example_complete}.

Without the methods of our paper the drawdown constraints are in general very hard to study and we are not aware of any works investigating them in the generality considered here. Nevertheless, they have received some attention in the financial literature, which one would expect given their practical significance. Magdon-Ismail and Atiya \cite{Magdon-IsmailAtiya:04} derived results linking the maximum drawdown to the mean return. Chekhlov, Uryasev and Zabarankin \cite{CUZ:05} analysed discrete-time portfolio optimisation where the investors maximises the expected return subject to risk constraints expressed in terms of drawdowns. They reduce the problem to a linear programming problem which can then be solved numerically. 

In the continuous time framework, apart from the early contributions in \cite{GZ} and \cite{CK}, drawdown constraints have recently been considered in setups with consumption. Roche \cite{Roche} investigated maximisation of expected utility of consumption over infinite time horizon for a power utility and under a linear drawdown constraint. Elie and Touzi \cite{ET} generalised this to a general class of utility functions in the setting of zero interest rates obtaining explicit representation of the solution. Subsequently, Elie \cite{Elie:08} analysed the problem of maximising the expected utility of consumption and terminal wealth on a finite time horizon. He did not have explicit formulae but rather represented the value function as the unique (discontinuous) viscosity solution of the Hamilton-Jacobi-Bellman equation. All the above works considered only the Black-Scholes market. It is not clear at present if, and to what extent, our methods extend to such problems.

Finally, we mention that Vecer \cite{Vecer:06} analysed options on drawdowns as a more effective way against portfolio losses than put or lookback options. Further analysis of option sensitivities to drawdowns was presented in Pospisil and Vecer \cite{PospisilVecer:10}.

The paper is organised as follows. Firstly, in Section \ref{sec:market_model}, we introduce the financial
market, give definitions and formulate the main portfolio optimisation problems of interest. In Section \ref{AY_processes}, we discuss the necessary results on Az\'ema-Yor processes. Section \ref{sec:main_num} presents the main result and its proof. It considers the problem with uniform units: the wealth in both the utility function and the drawdown constraint is discounted by the same numeraire. In Section \ref{sec:main_dollars}, we provide our results for utility of ``wealth in dollars" but subject to drawdown condition on the discounted wealth, as in \cite{CK,GZ}. This requires stronger asymptotic assumptions on $U$ and $w$ as well as deterministic interest rates. Section \ref{sec:logut} is devoted to the drawdown constrained optimisation problem with an asymptotically logarithmic utility. Finally, in Section \ref{sec:examples} several examples are presented. We first consider a general market model which admits price deflators as in Karatzas and Kardaras \cite{KarKar} and give sufficient conditions for 
finiteness of the value function. Then in Section \ref{sec:example_complete} we specialise to the complete market with deterministic coefficients and give explicit solutions, extending results in \cite{CK}. Finally, Section \ref{sec:example_incomplete} provides an explicit solution for an incomplete market model.

The Appendix contains some technical lemmas needed in the proofs but which are of independent interest. In particular we show continuity of the value function -- the long-term (asymptotic) growth rate of the expected utility of wealth -- in the utility function $U$ and its invariance under perturbation of $U$ on some initial interval $[0,x_0]$.

\section{Financial market model and portfolio choice problems}
\label{sec:market_model}

We consider a general financial market model with no frictions. 
The dynamics of traded assets are represented by a vector $\tilde S=(\tilde S^0,\tilde S^1,\ldots,\tilde S^d)$ of semimartingales defined on a filtered probability space $\pspace$ satisfying the usual conditions. We write $\ovl X_t:= \sup_{u\leq t} X_u$ for the running supremum of a process $X$. 
We say that a process $X$ is \emph{max-continuous} if $(\ovl X_t)_{t\geq 0}$ is a continuous process.

We fix a baseline asset, or a numeraire, $N=\tilde S^0$, and express all other quantities in units of $N$. The traded assets are then given as $S:=(1, S^1,\ldots, S^d)$, where $S_t^i=\tilde S^i_t/N_t$. It is customary for $N$ to be the (domestic) savings account but this is not necessary. In particular, the market could be spanned by $\tilde S^1,\ldots,\tilde S^d$ and $N$ could be the wealth process of some trading strategy in these assets. We assume that $N$ is max-continuous, strictly positive and $N_0=1$.
 Agents are allowed to invest by trading in the usual self-financing way as long as their wealth processes are \emph{max-continuous} and strictly positive. All wealth processes are expressed in units of $N$:
\begin{defn}\label{def:wealth}
An adapted semimartingale $(V_t)$ is called a \emph{wealth process} if it is strictly positive, $V_t>0$ and $V_{t-}>0$ for all $t\geq 0$ a.s., $V$ is max-continuous and there exists an $(\F_t)$--predictable process $\pi=(\pi^1,\ldots,\pi^d)$ such that $V_t=V_0+\int_0^t \sum_{i=1}^d \pi^i_u \td S^i_u$, where the (vector) integral is assumed to be well-defined. The set of wealth processes with $V_0=v_0$ is denoted $\A(v_0)$.
\end{defn}
We note that in general $\A(v_0)$ depends on our choice of the numeraire $N$ since max-continuity of $V$ and $VN/\tilde S^i$  are not necessarily equivalent. They are if, for example, all the assets $\tilde S^i$ are continuous. Likewise, if we consider the case when $\tilde S^i$ may only have jumps which are negative, totally inaccessible and unbounded, then $\A(v_0)$ is the same for all continuous numeraires and corresponds to wealth processes with no short selling restriction.

Note that so far we have not assumed ``no-arbitrage" in either strong sense of existence of an equivalent martingale measure $\Q$, or in a weaker sense of existence of a benchmark asset $\hat V$ for which all $V/\hat V$ are supermartingales, see Section \ref{sec:examples}. Neither have we made any strong assumptions on the integrability of $(\pi_t)$ which would make wealth processes $\Q$--martingales. Instead we consider utility maximisation in a general setting. As motivated in the Introduction, we are interested in enforcing drawdown constraints.
\begin{defn}\label{def:DDfunc}
We say that $w$ is a \emph{drawdown function} if it is non-decreasing and
\begin{equation}\label{eq:w-assumption}
\exists \alpha_1: 0 < w(x)/x \leq \alpha_1 < 1,\quad x> 0.
\end{equation}
We say that $(V_t)$ satisfies the \emph{$w$--drawdown} ($w$-DD) condition if
$$\min\{V_{t-},V_t\}>w(\sup_{u\leq t}V_u),\quad t\geq 0,\quad a.s.$$
The set of $(V_t)\in \A(v_0)$ which satisfy $w$-DD is denoted $\A^w(v_0)$.
\end{defn}
We stress that wealth is expressed in units of the baseline asset $N$ and hence the drawdown constraint is relative to our choice of $N$.
The canonical example of a drawdown function is: $w(x)=\alpha x$, see Example \ref{ex:lin_dd}. To the best of our knowledge this is the only example which has been considered in the literature, including \cite{CK,Elie:08,ET,GZ,Roche}. We consider a general possibly non-linear drawdown constraint. In particular, Definition \ref{def:DDfunc} allows for a piece-wise constant\footnote{More precisely, we can take $w$ to be piece-wise constant on $[v_0,\infty)$ for any $v_0>0$ and for the drawdown problem we only consider $w(x)$ for $x$ greater than the initial capital.} function $w$ which has the effect that drawdown constraint is updated discretely at times when the wealth process reaches a new threshold. Definition \ref{def:DDfunc} also allows for $w(x)\equiv c>0$ when $\A^w$ corresponds to wealth processes bounded below by a constant $c$.

One of the main aims of this paper is to relate portfolio choice problems with and without a drawdown constraint. The investor is assumed to be maximising the long-term growth rate of her expected utility of wealth. More precisely we consider the following two problems
\begin{prob}\label{pb:utmax_DD}
Given $v_0>0$, a drawdown function $w$ and a function $U$ compute
\begin{equation}\label{eq:opt_prob_DD}
\begin{split}
\cer_{U}^w(v_0) &:= \sup_{V\in \A^w(v_0)}\R_{U}(V),
\\ &\quad \textrm{where}\quad \R_{U}(V):=\limsup_{T\to \infty} \frac{1}{T}\log \e\left[U\left(V_T\right)\right],
\end{split}
\end{equation}
along with the optimal wealth process which achieves the supremum.
\end{prob}
Hereafter $\log$ is extended to $\re\setminus \{0\}$ via $\log(x)=-\log(-x)$ for $x< 0$.

\begin{prob}\label{pb:utmax}
Given $v_0>0$ and a function  $U$ compute
\begin{equation}\label{eq:opt_prob}
\begin{split}
\cer_{U}(v_0) &:= \sup_{X\in \A(v_0)}\R_{U}(X)
\end{split}
\end{equation}
along with the optimal wealth process which achieves the supremum.
\end{prob}

Problem \ref{pb:utmax} can be solved in number of fairly general setups, see Section \ref{sec:examples} below. 
Our aim is to build a direct link between the solution to this problem and the solution to a seemingly more complex Problem \ref{pb:utmax_DD} which features pathwise drawdown constraints. Problem \ref{pb:utmax_DD}, considered in Theorem \ref{thm:main} below, is a greatly generalised version of the problem introduced by Grossman and Zhou \cite{GZ} and analysed later by Cvitani\'c and Karatzas \cite{CK}. Firstly, we allow for an almost arbitrary utility function $U$ and not just the power utility.
Secondly, we consider a general possibly non-linear drawdown constraint. Finally, we work in a general semimartingale financial market model and not a complete Black-Scholes-like model. Working in such a generality we can not hope for an explicit solution to Problem \ref{pb:utmax_DD} as in \cite{CK,GZ}. 
Our main result offers second best: we obtain an explicit formula for the value function and the optimal investment strategy in Problem \ref{pb:utmax_DD} in terms of the value function and the optimal investment strategy in Problem \ref{pb:utmax} but with a suitably modified utility function.

The idea to look at $\R_U(V)$ -- the growth rate of the expected utility -- goes back to Dumas and Luciano \cite{DumasLuciano:91}, Grossman and Vila \cite{GrossmanVila:92} and Grossman and Zhou \cite{GZ}. 
CER above stands for \emph{Certainty Equivalent Rate} and is interpreted as the critical safe rate -- if the investor was offered such (or higher) rate of growth via other investment opportunities she would be happy to abandon the market and move to the alternative investment opportunities.
Similar criterion also appears in the \emph{risk-sensitive control} literature e.g.\ Fleming and Sheu \cite{FleSheu}, see also Section \ref{sec:example_incomplete}. The criterion is designed to capture the long-horizon optimality and is often more tractable\footnote{We note however that it may fail to provide strategies optimal on a finite time horizon, as discussed by Klass and Nowicki \cite{KlassNowicki:05} in the context of drawdown constraints.} than the fixed-horizon utility maximisation of terminal wealth, cf.~Guasoni and Robertson \cite{GuasoniRobertson:11}. 

Note that Problem \ref{pb:utmax_DD} has unified units: both the drawdown and the utility are applied to wealth in units of $N$. In \cite{CK,GZ} the drawdown is relative to $N$ (which is the savings account) but the reward functional is taken of the ``wealth in dollars": $\R_U(VN)$. This introduces further inhomogeneity and is solved in Section \ref{sec:main_dollars} under additional assumptions. Both Sections \ref{sec:main_num} and \ref{sec:main_dollars} consider $U$ which is either always positive or always negative. Utility functions with behaviour similar to logarithm are treated in Section \ref{sec:logut} where \eqref{eq:opt_prob_DD} is modified into \eqref{eq:opt_prob_log}.

Finally, we note that we consider only wealth processes which are strictly positive. However, in view of Lemma \ref{lem:2U}, this is not restrictive and we could allow wealth processes which become zero from some point onwards. Likewise, we impose strict drawdown condition but this could be relaxed. We could allow the drawdown constraint to be hit and the process would remain constant thereafter. Lemma \ref{lem:2U} and Theorem \ref{thm:main} show that such wealth processes can be excluded without affecting the value of problems of long-run utility maximisation considered in this paper.

\begin{remark}
As stated above, in \eqref{eq:opt_prob_DD}--\eqref{eq:opt_prob} and throughout the paper, we extend $\log$ to $\re\setminus \{0\}$ via $\log(x)=-\log(-x)$ for $x<0$. This implies that $\R_{U_1}(X)\geq \R_{U_2}(X)$ if $U_1\geq U_2$ are two functions of the same sign. Note also that $V_t\equiv v_0\in \A^w(v_0)\subset \A(v_0)$ and hence $\cer_{U}\geq \cer_U^w\geq 0$.
\end{remark}

\section{Drawdown constraints and Az\'ema--Yor processes} \label{AY_processes}
We recall now the so--called Az\'ema--Yor processes. We use their properties established in Carraro, El Karoui and Ob\l\'oj \cite{CEO} to build an explicit and model-independent bijection between $\A(v_0)$ and $\A^w(v_0)$. It will be our crucial tool used to relate Problems \ref{pb:utmax}, \ref{pb:utmax_DD} and their solutions. 

\begin{proposition}[Carraro, El Karoui and Ob\l\'oj \cite{CEO}]\label{prop:AY}
Let $F'$ be a locally bounded function, $F(x)=F(x_0)+\int_{x_0}^x F'(u)\td u$, and $(X_t)$ a max-continuous $(\F_t)$--semimartingale. The associated Az\'ema--Yor process $M^F(X)$ is given via
\begin{equation}\label{eq:AYdef}
M^F_t(X):=F(\ovl X_t)-F'(\ovl X_t)(\ovl X_t - X_t) = F(X_0)+\int_0^t F'(\ovl X_u)\td X_u,\ t\geq 0,
\end{equation}
where $\ovl X_t := \sup_{u\leq t} X_u$. Further
\begin{itemize}
\item if $F'\geq 0$ then $\ovl{M^F(X)}_t=F(\ovl X_t)$, $t\geq 0$,
\item if $F'>0$ then $M^K(M^F(X))=X$ with $K=F^{-1}$ the inverse of $F$,
\item if $F$ is concave then $M^F_t(X)\geq F(X_t)$, $t\geq 0$.
\end{itemize}
\end{proposition}
The above combines Definition 2.1 and Proposition 2.2 in \cite{CEO} while the last property is clear (see also Proposition 4.12 point c) in \cite{CEO}). 
It is important to observe that Az\'ema--Yor processes automatically satisfy a drawdown property. More precisely, consider $M^F(X)$ for $X\geq 0$ and $F'>0$. Then from \eqref{eq:AYdef}, using first $F'(\ovl X_t)X_t\geq 0$ and then $F(\ovl X_t)=\ovl{M^F(X)}_t$, we obtain 
$$M^F(X)_t\geq F(\ovl X_t)-F'(\ovl X_t)\ovl X_t = w(\ovl{M^F(X)}_t),\quad t\geq 0,$$
where $w(x)=x-K(x)/K'(x)$, $K:=F^{-1}$. The following result shows that we can start with $w$, solve the ODE for $K$ and hence obtain $F$. It builds an explicit and model-independent bijection between $\A(v_0)$ and $\A^w(v_0)$ and will be our main tool in this paper.  
\begin{proposition}\label{prop:DD}
Let $w$ be a function satisfying \eqref{eq:w-assumption} and define
\begin{equation}\label{eq:Kdef}
K_w(x) := v_0 \exp \left( \int_{v_0}^x \frac{1}{u - w(u)} \td u \right), x \geq v_0>0,
\end{equation}
which is continuous and strictly increasing and has a well defined inverse $F_w:=K_w^{-1}:[v_0,\infty)\to [v_0,\infty)$.\\
The mapping $V\to M^{F_w}(V)$ is a bijection between $\A(v_0)$ and $\A^w(v_0)$ with its inverse given by $X\to M^{K_w}(X)$.\\
 For $V\in \A(v_0)$ we have $X:= M^{F_w}(V)\in \A^w(v_0)$ satisfies
\begin{equation}\label{eq:SDEwDD}
\td X_t = \left(X_{t-}-w\left(\ovl{X}_t\right)\right)\frac{\td V_t}{V_{t-}},\ t\geq 0.
\end{equation}
Finally, if $w$ is nondecreasing then $K_w$ is convex and $F_w$ is concave.
\end{proposition}
\begin{proof}
First note that in Proposition \ref{prop:AY} to define $M^F(X)$ it suffices that $F$ is defined on the range of possible values of $\ovl X$. 
In particular, both $M^{F_w}(V)$ and $M^{K_w}(X)$ in the statement of Proposition \ref{prop:DD} are well defined.

Properties of $K_w$ follow by a straightforward differentiation. Observe that by \eqref{eq:w-assumption} $K_w(x)<\infty$ for all $x>v_0$ and 
$K_w(\infty)=\infty$ so that the inverse $F_w$ is well defined on $[v_0,\infty)$ and $F_w(\infty)=\infty$. Further $F_w$ is increasing so, by Proposition \ref{prop:AY}, if $V\in \A(v_0)$ then the process $X=M^{F_w}(V)$ is well defined, strictly positive and max-continuous. We apply Theorem 3.4 in \cite{CEO}  with $v^*_0 = v_0$. Note that, in the notation therein, we have $\zeta=\infty$ a.s. since $r_w=\infty$ and $K_w(\infty)=\infty$. We conclude that $X$ satisfies the required $w$--drawdown property for all $t\geq 0$ and \eqref{eq:SDEwDD} holds, and hence that $X\in\A^w(v_0)$.

Likewise, for the converse, $X\in \A^w(v_0)$ satisfies the $w$-DD condition for all times and hence we apply Theorem 3.4 in \cite{CEO} with $\zeta=\infty$. 
Finally, the last statement is clear since when $w$ is non-decreasing then, by direct differentiation, so is $K'$.\hfill$\square$
\end{proof}
Note that we have $w(x)=x-K_w(x)/K_w'(x)$, $x\geq v_0$, a property which we will use often below. Strictly speaking, in view of \eqref{eq:Kdef}, this holds for $x>v_0$ and is used to define $K_w'(v_0)$. Likewise we take $F_w'(v_0)=1/K_w'(v_0)$. We adopt this convention hereafter.

By \eqref{eq:AYdef}, in the above it is sufficient to consider $F_w(v)$ for $v\geq v_0$ since $\ovl{V}_t\geq V_0=v_0$. We are free to define $F_w$ on $[0,v_0)$ in any way without affecting $X_t$. As we will see later, any extension which preserves the sign, monotonicity and concavity of $F_w$ will be allowed.
For completeness we specify one such extension by extending $F_w$ for all positive $v$ as follows
\begin{equation}\label{eq:Fdef}
F_w(v) := \left\{\begin{array}{ll} K_w^{-1}(v) & \textrm{if $v \geq v_0$}\\
F_w'(v_0+)(v-v_0) + v_0 & \textrm{if $0 \leq v <v_0$} \end{array} \right.
\end{equation}
so that $F_w(0)=w(F_w(v_0))=w(v_0)>0$ and $F_w$ is increasing and concave on $[0,\infty)$ if $w$ is nondecreasing. 

\begin{example}\label{ex:lin_dd}
Consider a linear drawdown constraint $w(x)=\alpha x$, $\alpha\in (0,1)$. Then $K_w(v)=v_0(v/v_0)^{1/(1-\alpha)}$ and $F_w(x)=v_0(x/v_0)^{1-\alpha}$. For $V\in \A(v_0)$, Proposition \ref{prop:DD} gives us that $X:=M^{F_w}(V)\in \A^w(v_0)$ and an explicit calculation using \eqref{eq:AYdef} gives
$$X_t=M^{F_w}_t(V)=v^{\alpha}_0(\alpha \overline{V}_t^{1-\alpha}+(1-\alpha)\overline{V}_t^{-\alpha}V_t) ,\ t\geq 0,$$
with an analogous expression for $V$ in terms of $X$.

\end{example}

\begin{example}\label{ex:const_dd}
As an extreme example, consider a constant drawdown constraint $w(x)\equiv c \in (0,v_0)$. Then $K_w(v)=\frac{v_0}{v_0-c}v-\frac{cv_0}{v_0-c}$ and $F_w(x)=c+\frac{v_0-c}{v_0}x$. It follows that for $V\in \A(v_0)$ we have
$$X_t=M^{F_w}_t(V)=\frac{v_0-c}{v_0}V_t+c>c.$$
\end{example}
 
The methodology based on Az\'ema--Yor processes, as introduced above, is perfectly suited for our analysis. We note however that it is also possible to study drawdowns, and in particular laws related to the first time a drawdown occurs, by considering $X_t=(X_t-\ovl X_t)+\ovl X_t$ and applying methods of processes of class $\Sigma$, see e.g.\ Cheridito, Nikeghbali and Platen \cite{CNP}.

\section{Main results}\label{sec:main_num}
We are now ready to formulate our main results. The essence of the results is simple and explicit: the $w$--drawdown problem with a utility function $U$ has the same value as the unconstrained problem with the utility function $U\circ F_w$: $\cer^w_{U}=\cer_{U\circ F_w}$, where $w$ and $F_w$ are related by \eqref{eq:Kdef} and \eqref{eq:Fdef}. Further, the optimal wealth process is given by $ M^{F_w}(V^*)$, where $V^*$ is the optimal wealth for the unconstrained problem. Note that these results do not depend on the underlying stochastic market model.
We impose
\begin{asmp}\label{ass:U}
Assume that, for some $\varepsilon>0$, $U$ satisfies either
\begin{equation*}
\frac{U(x)}{x^{\varepsilon}} \vax \infty,\quad \textrm{and $U$ is strictly positive on }(0,\infty),
\end{equation*}
or
\begin{equation*}
 U(x)x^{\varepsilon} \vax 0, \quad \textrm{and $U$ is strictly negative on }(0,\infty).
\end{equation*}

\end{asmp}
This insures that our utility functions are of constant sign and they dominate a power utility. We will further assume that they admit positive finite Asymptotic Elasticity in the sense of Kramkov and Schachermayer \cite{KrSch}. Throughout the paper a \emph{utility function} simply means a non-decreasing concave function and $U'_+$ denotes its right derivative, $U'_+(x) := \lim_{y \downarrow x} \frac{U(y) - U(x)}{y-x}$. Recall Definition \ref{def:wealth}, and in particular that wealth processes are \emph{max-continuous}, and Definition \ref{def:DDfunc}.

\begin{theorem}\label{thm:main} Let $w$ be a drawdown function, $v_0>0$ and $U$ a utility function satisfying Assumption \ref{ass:U} and 
\begin{equation}\label{eq:AEU}
\limsup_{x\to \infty} \frac{xU'_+(x)}{|U(x)|}<\infty.
\end{equation}
Recall that $K_w$ is given by \eqref{eq:Kdef} and let $F_w$ be its inverse extended to $[0,\infty)$ as in \eqref{eq:Fdef} or in any other way which preserves monotonicity, concavity and $F_w(0)>0$. Assume that, for some $\delta>0$, $\cer_{G}(v_0)<\infty$ where $G(x)=U \circ F_{w}(x)$ when $U<0$ and $G(x)=(U \circ F_{w}(x))^{1+\delta}$ when $U>0$.
Then
\begin{equation*}
\cer^w_{U}(v_0)=\cer_{U\circ F_w}(v_0)<\infty
\end{equation*}
and if $V^*\in  \A(v_0)$ achieves the maximum in the unconstrained problem then $M^{F_w}(V^*)\in  \A^w(v_0)$ achieves the maximum in the $w$-drawdown constrained problem.
\end{theorem}

\begin{proof} For simplicity, in the proof we write $K=K_w$ and $F=F_w$.
Let $V\in \A(v_0)$ and $X:= M^F(V)$ which is in $\A^w(v_0)$ by Proposition \ref{prop:DD}.
Now, by Proposition \ref{prop:AY} we obtain directly
\begin{equation}\label{eq:first_ineq}
U\left(X_t\right)=U\left(M^F_t\left(V\right)\right) \geq U \left( F\left(V_t\right) \right),\ t\geq 0,
\end{equation}
which readily implies $\R_{U}(X)\geq \R_{U\circ F}(V)$. Taking supremum over all $V \in \mathcal{A}(v_0)$ we conclude
\begin{equation*}
 \cer^w_{U}(v_0) \geq \cer_{U \circ F}(v_0)\ .
\end{equation*}
It follows also that if we had equality and the right hand side was attained by a wealth process $V^*$ then the left hand side is attained by $ M^F(V^*)$, as required.

It remains to establish the reverse inequality. The idea of the proof is to consider a sequence of unconstrained problems whose value functions all dominate $\cer^w_{U}(v_0)$ and converge to $\cer_{U \circ F}(v_0)$. We will do this by relaxing the drawdown constraint $w$ to $w_n$ and considering utility functions $U\circ F_{w_n}$. This will allow us to obtain \eqref{eq:second_ineq} below which reverses the inequality in \eqref{eq:first_ineq}. Let
\begin{equation}\label{eq:wn_and_w}
w_n(x):=(1+\frac{1}{n})w(x)-\frac{1}{n}x= w(x)-\frac{1}{n}\frac{K(x)}{K'(x)}=
x-(1+\frac{1}{n})\frac{K(x)}{K'(x)}, \quad x \geq v_0,
\end{equation}
where the equalities follow from $w(x)=x-K(x)/K'(x)$. We take $n$ large enough so that $w_n(x)$ satisfies \eqref{eq:w-assumption} but we note that $w_n$ may fail to be globally non-decreasing on $(0,\infty)$.

It follows by a direct computation that
$$K_n(x):= K_{w_n}(x)= v_0 \exp \left( \int_{v_0}^x \frac{1}{u - w_n(u)} \td u \right) = v_0^{\frac{1}{1+n}}\left(K(x)\right)^{\frac{n}{1+n}},\quad x\geq v_0,$$
and $w_n(x)=x-K_n(x)/K'_n(x)$, where by our convention this is used to define $K'_n(v_0)$.
Consider $(X_t) \in \A^w(v_0)$ and let $Y^n_t:=M^{K_n}_t(X)$ which is an element of $\A(v_0)$ by Proposition \ref{prop:DD}. Using \eqref{eq:AYdef} and the drawdown property of $X$ we obtain
\begin{equation*}
Y^n_{t} \geq K_n\left(\overline{X}_t\right) - K'_n \left(\overline{X}_t\right) \left(\overline{X}_t - w\left(\overline{X}_t\right) \right) = h_n\left(\overline{X}_t\right), \label{eq:y_n>x}
\end{equation*}
where $h_n(x)$ is defined for $x \geq v_0$ as
\begin{equation*}
\begin{split}
h_n(x) :&= K_n(x)-K'_n(x)(x-w(x))=K_n(x)\left(1-\frac{K'_n(x)}{K_n(x)}(x-w(x))\right)\\
&=K_n(x) \left( 1 - \frac{x-w(x)}{x-w_n(x)}\right)=K_n(x)\left(1-\frac{1}{1+1/n}\right)\\
&=\frac{1}{1+n}K_n(x),
\end{split}
\end{equation*}
where we used \eqref{eq:wn_and_w} and $w_n(x) = x - K_n(x)/K'_n(x)$.

Let $F_n(v)$ be the inverse of $K_n(v)$, for $v\geq v_0$ and extended to $[0,\infty)$ via \eqref{eq:Fdef}. Explicitly, we have $F_n(v)=F\left(v_0^{-1/n}v^{\frac{1+n}{n}}\right)$, $v \geq v_0$ and $F_n(v)=F(v)- \frac{1}{n}(F(v_0)-F(v))$, $v\in [0,v_0)$. $F_n$ is continuous and strictly increasing on $[0,\infty)$ and we take $n$ large enough so that $F_n(0)>0$. 

Assume, which we will argue later, that $U\circ F_n$ satisfies \eqref{eq:AEU} and hence we may apply Lemma \ref{lemma:ae_u}. Setting $x^n_0=v_0/(n+1)$ we deduce that there exists $\gamma_n\in \re$, $\gamma_n\neq 0$ such that for all $x\geq x^n_0$ and all $\lambda\geq 1$ we have $U\circ F_n(\lambda x)\leq \lambda^{\gamma_n} U\circ F_n(x)$. We use this with $\lambda=(1+n)$ and $x=K_n\left(\overline{X}_t\right)/\lambda$. Observe that $\overline{X}_t \geq v_0$ and that $K_n(v_0) = v_0$ so $x\geq x^n_0$ as required. Combining all the above, we obtain
\begin{equation}\label{eq:second_ineq}
\begin{split}
U \circ F_n \left(Y^n_t \right) & \geq U\circ F_n\left(h_n\left(\overline{X}_t\right)\right) = 
 U \circ F_n \left(\frac{1}{1+n} K_n \left( \overline{X}_t\right) \right)\\&  \geq  \left(\frac{1}{1+n}\right)^{\gamma_n} U \left(\overline{X}_t\right) \geq  \left(\frac{1}{1+n}\right)^{\gamma_n} U \left(X_t\right).
\end{split}
\end{equation}
The factor of $(1+n)^{-\gamma_n}$ disappears when we apply $\frac{1}{t}\log$ and let $t\to\infty$:
\begin{equation}\label{eq:reward_ineq_reverse}
\R_{U\circ F_n}(Y^n)\geq \R_U(X).
\end{equation}
Taking supremum over $X\in \A^w(v_0)$ we conclude that
$$\cer_{U\circ F_n}(v_0)\geq \cer^w_{U}(v_0)$$
and thus we indeed have a sequence on unconstrained problems with value functions all dominating $\cer^w_{U}(v_0)$. 

Assume, which we will show later, that $U\circ F$ and $U\circ F_n$ satisfy Assumption \ref{ass:U} and now argue the convergence of $\cer_{U\circ F_n}(v_0)$ to $\cer_{U\circ F}(v_0)$ using Lemma \ref{lem:cnv_cer} in the Appendix.
For $v\geq v_0$ we have $F(v)\leq F_n(v)=F\left(v_0^{-1/n}v^{\frac{1+n}{n}}\right)$
and for $v\in [0,v_0)$ we have $c_n F(v)\leq F_n(v)\leq F(v)$ where $c_n = 1 + \frac{1}{n}\frac{w(v_0)-v_0}{w(v_0)}$ and we take $n>n_0$ chosen such that $c_n>c_0:=c_{n_0} >0$. Fix $\epsilon\in (0,\min\{1,v_0\})$ and note that $v_0^{-1/n}\leq \epsilon^{-1/n}<\epsilon$. Together the above give us
\begin{equation}\label{eq:FFnbound}
c_0F(v)\leq F_n(v)\leq F(v^{1+1/n}/\epsilon),\quad v\geq \epsilon.
\end{equation}
Thanks to \eqref{eq:AEU}, we can apply Lemma \ref{lemma:ae_u} to $U$ with $x_0=c_{0}F(\epsilon)$ to see that there exists a non-zero $\gamma'\in \re$ such that $U(\frac{1}{c_0} c_0F(v))\leq c_0^{-\gamma'}U(c_0F(v))$, for all $v\geq \epsilon$. We thus obtain
\begin{equation*}
c_0^{\gamma'}U\circ F(v)\leq U(c_0 F(v))\leq U\circ F_n(v)\leq U\circ F\left(\frac{1}{\epsilon}v^{1+1/n}\right),\quad v\geq \epsilon.
\end{equation*}
Together with $\cer_{G}<\infty$ and Lemma \ref{lem:2U}, Lemma \ref{lem:cnv_cer} now yields
$$ \cer_{U\circ F_n}(v_0) \van \cer_{U\circ F}(v_0)<\infty.$$

It remains to argue the properties assumed above: that $U\circ F_n$ satisfies \eqref{eq:AEU} and that $U\circ F$ and $U\circ F_n$ satisfy Assumption \ref{ass:U}. Observe that for $x > v_0$:
\begin{equation}\label{eq:FAEbound}
x\frac{F_n'(x)}{F_n(x)} = \frac{x}{K'_n(F_n(x))F_n(x)} = \frac{F_n(x) -w_n(F_n(x))}{F_n(x)} \leq 1
\end{equation}
and hence
\begin{equation*}
 x \left(U \circ F_n(x)\right)'_+ = F_n(x) U'_+(F_n(x)) \frac{xF_n'(x)}{F_n(x)} \leq F_n(x) U'_+(F_n(x)). 
 \end{equation*}
We conclude that
\begin{equation*}
\limsup_{x\to \infty}\frac{ x \left(U \circ F_n(x)\right)'_+}{|U\circ F_n(x)|}\leq \limsup_{x\to \infty}\frac{F_n(x) U'_+(F_n(x))}{|U\circ F_n(x)|} =
\limsup_{y\to \infty}\frac{ y U'_+(y)}{|U(y)|},
\end{equation*}
where we used the fact that $F$ is a strictly increasing continuous map of $[v_0,\infty)$ onto itself so any sequence $y_m\to \infty$ can be represented as $y_m=F_n(x_m)$ for $x_m=K_n(y_m)\to \infty$ as $m\to \infty$.

Finally, from \eqref{eq:w-assumption} note that $(x/v_0)\leq K(x)/v_0\leq (x/v_0)^{1/(1-\alpha_1)}$, $x\geq v_0$. Suppose $U>0$ and recall $\varepsilon$ from Assumption \ref{ass:U}. Take $\delta$ small enough so that $\delta/(1-\alpha_1)<\varepsilon$. Then, using $K(F(x))=x$,
\begin{equation*}
\begin{split}
\frac{U\circ F(x)}{x^\delta} &= \frac{U\circ F(x)}{F(x)^\delta}\left(\frac{F(x)}{K(F(x))}\right)^\delta \geq  \frac{U\circ F(x)}{F(x)^\delta}v_0^{\delta/(1-\alpha_1)} F(x)^{\delta (1-1/(1-\alpha_1))}\\
& \geq \frac{U\circ F(x)}{F(x)^{\delta/(1-\alpha_1)}}v_0^{\delta/(1-\alpha_1)} \to \infty,\quad \textrm{as }x\to \infty.
\end{split}
\end{equation*}
We conclude that $U\circ F$ satisfies Assumption \ref{ass:U}. 
Note that we used here $\alpha_1\in (0,1)$. For $U<0$ we have a similar argument which uses $K(x)\geq x$. Finally the same arguments apply to $U\circ F_n$ since $w_n$ satisfy \eqref{eq:w-assumption}.\hfill$\square$
\end{proof}
From Lemma \ref{lem:2U} we immediately have
\begin{corollary}\label{coro:initial_wealth}
Under the assumptions of Theorem \ref{thm:main} we have for any, possibly random, bounded and bounded away from zero initial capital $\nu$ 
\begin{equation*}
\cer^w_{U}(\nu)=\cer^w_{U}(1)=\cer_{U\circ F_w}(1)=\cer_{U\circ F_w}(\nu)<\infty
\end{equation*}
and if $V^*\in  \A(1)$ achieves $\cer_{U\circ F_w,}(1)$ then $\nu V^*$ achieves $\cer_{U\circ F_w}(\nu)$ and $M^{F_w}(\nu V^*)\in  \A^w(\nu)$ achieves $\cer^w_{U}(\nu)$.
\end{corollary}

We close this section with a series of remarks about the above results, their assumptions and proofs.
\begin{remark}
From the proof of Theorem \ref{thm:main}, and in particular from Lemma \ref{lem:cnv_cer} in the Appendix, it is clear that when $U<0$ we have in fact $\cer^w_{U}<\infty$ if and only if $\cer_{U\circ F_w}<\infty$.
\end{remark}

\begin{remark}
We note first that for $U\geq 0$ concavity implies \eqref{eq:AEU}: $U(x)\geq U(x)-U(0)\geq xU'_+(x)$ so that $\limsup_{x \rightarrow \infty}\frac{xU'_+(x)}{U(x)} \leq 1$. Assumption \ref{ass:U} implies that this limit is also strictly positive. Indeed, we have $\frac{U(x)}{x^\varepsilon} \rightarrow \infty$, for some $\varepsilon >0$, which implies that there exists a sequence $x_k\to \infty$ such that $\left(\frac{U(x)}{x^\varepsilon}\right)'_+|_{x=x_k} = \frac{U(x_k)}{x^{\varepsilon+1}_k} \left(\frac{x_k U'_+(x_k)}{U(x_k)} - \varepsilon\right)>0$. Thus, $\limsup_{x \rightarrow \infty}\frac{xU'_+(x)}{U(x)} \geq \varepsilon$. 

However \eqref{eq:AEU}, even if we assume the limit is strictly positive, does not imply Assumption \ref{ass:U}. This can be seen by considering a $U$ which alternates between linear and log-like behaviours. Specifically, consider $U_0$ with $U_0(0)=0$, $U'_{0}(0+)=1$, $U_0'(x)$ being continuous, with $U_0'(x) = U'_0(x_{2k})$ for $x\in[x_{2k},x_{2k+1}]$ and $U'_0(x)=(x-x_{2k+1}+U'_0(x_{2k+1})^{-1})^{-1}$ for $x\in [x_{2k+1}, x_{2k+2}]$ where $x_0=0$, $x_1=1$ and
\begin{equation}
\begin{split}
x_{2k}&:=\inf\left\{x \geq x_{2k-1} + 1: U_0(x)\leq x^{1/k}\right\}, \quad k\geq 1,\\
x_{2k+1}&:=\inf\left\{x \geq x_{2k} + 1: \frac{xU_0'(x)}{U_0(x)}\geq 1 - 1/k\right\},\quad k\geq 1.
\end{split}
\end{equation}
Choice of $x_i$ guarantees that $U_0$ does not dominate any positive power of $x$ asymptotically, but has a strictly positive asymptotic elasticity. More precisely, the asymptotic elasticity is equal to one as it is bounded from above by 1 since $U$ is a positive utility function and $\lim_{k\rightarrow \infty} \frac{x_{2k+1} U'(x_{2k+1})}{U(x_{2k+1})} = 1$.

For $U<0$ we need both \eqref{eq:AEU} and Assumption \ref{ass:U}. The latter implies that the limit in \eqref{eq:AEU} is non-zero but it could be infinite, as for $U(x)=-\mathrm{e}^{-x}$ which satisfies Assumption \ref{ass:U}. Conversely, $-\frac{1}{U_0(x)}$ is a utility function which satisfies \eqref{eq:AEU} but not Assumption \ref{ass:U}. 
\end{remark}

\begin{remark}\label{remark:U_shift} Observe that $\cer$ in Problems \ref{pb:utmax} and \ref{pb:utmax_DD} are invariant under a multiplication of $U$ by a positive constant. Further, for a positive utility function $U$, they are invariant under a constant shift of $U$ which preserves the sign. More precisely, write $C=\cer_{U}(v_0)$ and let $\kappa>0$. For any $\delta>0$, $V\in \A(v_0)$, taking $T$ large enough we have
$$\log \e[U(V_T)+\kappa]\leq \log(\mathrm{e}^{T(C+\delta)}+\kappa)\leq \log(2\mathrm{e}^{T(C+\delta)}).$$
This yields $\R_{U+\kappa}(V)\leq C+\delta$ and letting $\delta\searrow 0$ we have $\cer_{U+\kappa}(v_0)=C$.  Finally, both problems are invariant under changes of $U$ in the neighbourhood of zero. This is clear under the drawdown constraint since $U$ is never evaluated on $x\in (0,w(v_0))$. By Theorem \ref{thm:main} it is also true for the unconstrained problem. A direct argument for this is given in the proof of Lemma \ref{lem:2U} in the Appendix.
\end{remark}

\begin{remark}
By the opening arguments in the proof of Theorem \ref{thm:main}, if $V\in \A(v_0)$ then $\R_U(X)\geq \R_{U\circ F}(V)$, where $X=M^F(V)\in \A^w(v_0)$. In particular if $V$ is nearly optimal, say $\R_U(V)\geq \cer_{U\circ F_w}(v_0)-\epsilon$, then so is $X$: $\R_U(X)\geq \cer_{U}^w(v_0)-\epsilon$.
\end{remark}

\begin{remark}
We stress that throughout, similarly to  \cite{GZ,CK}, drawdown constraint is imposed on the wealth expressed in units of $N$. In practice this may be adequate if the so-called hurdle rate is present, see Guasoni and Ob\l\'oj \cite{GuasoniObloj:12}. However in many situations one is interested in avoiding drawdowns for the actual wealth process ``in dollars", $XN$. Suppose for simplicity that all the assets are continuous and let $V$ be a positive wealth process, $\td V_t=\sum_{i=1}^d \pi_t^i \td S^i_t$. Using our methodology we could consider $X=M^F(VN)/N$ which would indeed satisfy $XN\geq w(\ovl {XN})$, but $X$ would not be a wealth process of a self-financing portfolio in Definition \eqref{def:wealth}.

More specifically, suppose $N_t =\exp(\int_0^t R_u \td u)$ is the savings account, where $R$ is some positive adapted process. Let $D_t=1/N_t$.  Then, using \eqref{eq:AYdef},
\begin{eqnarray*}
\td \left( D_t M^F_t(VN)\right) &=& D_t F'(\ovl{VN}_t) \td (V_tN_t) + M^F_t(VN) \td D_t \\
& = & F'(\ovl{VN}_t) \td V_t +  D_t V_t F'(\ovl{VN}_t) \td N_t + M^F_t(VN) \td D_t\\
& = & - R_t \left(F(\ovl{VN}_t) - F'(\ovl{VN}_t) \ovl{VN}_t\right) D_t \td t + F'(\ovl{VN}_t) \sum_{i=1}^d \pi^i_t \td S^i_t.
\end{eqnarray*}
It follows that $M^F(VN)$ is the dollar value of a wealth process of a consumption and investment strategy, where the rate of instantaneous consumption is $R_t \left(F(\ovl{VN}_t) - F'(\ovl{VN}_t) \ovl{VN}_t\right)$. 

The above calculation suggests that drawdown constraints imposed on the undiscounted wealth should be considered in the context of maximisation of utility of consumption. We note however that it is not clear how to build bijection akin to Proposition \ref{prop:DD} above and what the adequate sets of wealth processes should be. We believe this is a challenging topic for further studies.
\end{remark}

\section{Utility of wealth in dollars}\label{sec:main_dollars}
We turn now to the inhomogeneous problem as considered by Grossman and Zhou \cite{GZ}. Assume that $N$ represents the savings account and therefore $V$ in Definition \ref{def:wealth} represents discounted wealth process. We seek to maximise the utility of ``wealth in dollars" $U(VN)$, but the drawdown constraint is imposed on the discounted wealth process $V$. Note that in the case of a linear constraint, $w(x)=\alpha x$, this is equivalent to saying that the drawdown constraint is imposed on $VN$ but is growing at a hurdle rate equal to the riskless rate\footnote{Hurdle rate $r$ means that if a new drawdown constraint is set at time $t$ then it grows at the rate $r$ for $u>t$ until a new constraint level is achieved, see Guasoni and Ob\l\'oj \cite{GuasoniObloj:12} for more details.}. In analogy to Problems \ref{pb:utmax}, \ref{pb:utmax_DD} we define
\begin{equation}\label{eq:utmax_dollars}
    \cerd_{U}(v_0):=\sup_{V\in \A(v_0)}\R_U(VN),\quad \cerd_U^w(v_0):=\sup_{X\in \A^w(v_0)}\R_U(XN).
\end{equation}
In order to be able to relate these problems we essentially need to go back to the homogenous case when $\R_U(VN)$ is replaced by $\R_U(V)$ and for this we need to be able to factor the discounting in and out of the reward functional $\R_U$. This is possible when $U$ is a power utility, $w$ is linear and $N$ is deterministic as in \cite{CK} and \cite{GZ}. Here we need to assume this holds asymptotically.
\begin{asmp}\label{ass:UwS0}
Assume the following three conditions hold
\begin{itemize}
  \item[(i)] $U$ is either strictly positive or strictly negative on $(0,\infty)$ and the following limit exists  $\frac{x U'_+(x)}{U(x)} \to \gamma \in (-\infty,1)\setminus \{0\}$ as $x\to \infty$,
  \item[(ii)] the following limit exists $\frac{w(x)}{x} \to \alpha \in (0,1)$ as $x\to \infty$,
  \item[(iii)] $N$ is increasing, deterministic and the following exists $r^*:=\lim_{T\to\infty} \frac{\log N_T}{T}$.
\end{itemize}
\end{asmp}
The first assumption is a strengthened version of the finite asymptotic elasticity of Kramkov and Schachermayer \cite{KrSch} which we assumed earlier in \eqref{eq:AEU}. It follows from Lemma \ref{lem:Uasel} in the Appendix that it implies Assumption \ref{ass:U} holds.
The second condition above is in fact equivalent to saying that $K_w$ in \eqref{eq:Kdef} also has such (converging) finite asymptotic elasticity. This is immediate since $xK_w'(x)/K_w(x)=x/(x-w(x))$.
We denote the CRRA (power) utility with $\powut{p}(x)=\frac{x^p}{p}$, $p\leq 1$. We assume $p\neq 0$, which is the case of logarithmic utility treated below in Section \ref{sec:logut}. Finally, we denote $w_\alpha(x)=\alpha x$ the linear drawdown function.

We assumed above that the interest rates are deterministic and that asymptotically $U$ is a power utility and $w$ is linear. Comparing with the setup in Section \ref{sec:main_num} these are strong assumptions and our main contribution, relative to \cite{CK}, is that we work in a general max-continuous semimartingale market.
\begin{theorem} \label{thm:main_dollars} Let $U$ be a utility function and $w$ a drawdown function for which Assumption \ref{ass:UwS0} holds.
Assume further that $\cerd_{\powut{\gamma(1-\alpha)(1+\delta)}}(v_0)<\infty$ for some $\delta>0$. Then
\begin{equation*}
\cerd^w_{U}(v_0)=\cerd^{w_\alpha}_{\powut{\gamma}}(v_0)=\cerd_{\powut{\gamma(1-\alpha)}}(v_0)+|\gamma|\alpha r^*<\infty
\end{equation*}
and  if $V^*\in \A(v_0)$ achieves the maximum in the unconstrained problem then $M^{F_w}(V^*)\in \A^w(v_0)$ achieves the maximum in the $w$-drawdown constrained problem, where $F_w$ is as in Theorem \ref{thm:main}.
\end{theorem}
\begin{remark}
Theorem 1.1 of Grossman and Zhou \cite{GZ} and Theorem 5.1 of Cvitani\'c and Karatzas \cite{CK} are consequences of the above statement. Namely, they specialise to $w=w_\alpha$, $U=\powut{\gamma}$ with $\gamma\in (0,1)$ and a particular (deterministic, constant coefficients) market setup. Standard techniques allow then to compute explicitly $\cerd_{\powut{\gamma(1-\alpha)}}(v_0)$ and find the optimal wealth process $V^*$, see Section \ref{sec:example_complete}. Therein we also discuss how various objects in \cite{CK} and in our paper relate explaining how methods of \cite{CK} helped us develop intuition behind this paper.
\end{remark}
\begin{proof}
Note that, by considering $V\equiv v_0$ and using $N_t\geq N_0=1$, we see that $\cerd_U\geq \cerd^w_U\geq 0$. Recall also that $\R_{U_1}(XN)\leq \R_{U_2}(XN)$ for two functions $U_1\leq U_2$ of the same sign.

Consider $X\in \A^w(v_0)$ and a small $\varepsilon>0$. As $N_tX_t\geq N_tw(v_0)\geq w(v_0)>0$, we can apply Lemma \ref{lem:Uasel} in the Appendix to obtain
\begin{equation}\label{eq:proof2eq1}
\R_{U}(NX)\leq \R_{\powut{\gamma(1+\varepsilon)}}(NX) = \R_{\powut{\gamma(1+\varepsilon)}}(X)+|\gamma|(1+\varepsilon)r^*,
\end{equation}
the last equality following since $N$ is deterministic.

Recall $K_w$ defined in \eqref{eq:Kdef}. Note that $xK_w'(x)/K_w(x)=x/(x-w(x))$ and hence that Assumption \ref{ass:UwS0} implies
$$\lim_{x\to \infty}\frac{xK_w'(x)}{K_w(x)}=\lim_{x\to\infty}\frac{x}{x-w(x)}=\lim_{x\to\infty}\frac{1}{1-w(x)/x}=\frac{1}{1-\alpha}.$$
$F_w:[v_0,\infty)\to [v_0,\infty)$ is the inverse of $K_w$ extended to $[0,\infty)$ through \eqref{eq:Fdef}. It is increasing, concave and we have $\lim_{x\to\infty}xF_w'(x)/F_w(x)=1-\alpha$.
Lemma \ref{lem:Uasel} then implies that $F_w(x)\leq c_+ \powut{(1-\alpha)(1+\varepsilon)}(x)$ for some $c_+\geq 1$ and all $x\geq v_0/2$. 
In consequence, for any $Y\in \A(v_0)$ with $Y\geq v_0/2$ 
\begin{equation}\label{eq:proof2eq2}
\begin{split}
\R_{\powut{\gamma(1+\varepsilon)}\circ F_w}(Y)& \leq \R_{\powut{\gamma(1-\alpha)(1+\varepsilon)^2}}(Y) \\&
=\R_{\powut{\gamma(1-\alpha)(1+\varepsilon)^2}}(NY)-|\gamma|(1-\alpha)(1+\varepsilon)^2r^*,
\end{split}
\end{equation}
which is finite for $\varepsilon$ small enough by assumption. Note that by Lemma \ref{lem:2U} it is sufficient to consider only such $Y$.

By a similar reasoning and using the assumption of the theorem we conclude that $\cer_{G}(v_0)<\infty$ for $G(x)=(\powut{\gamma(1+\varepsilon)}\circ F_w(x))^{1+\kappa}$ with $\kappa=\varepsilon\mathbf{1}_{\gamma>0}$ and for $\varepsilon$ small enough. Applying Theorem \ref{thm:main} and combining \eqref{eq:proof2eq1} with \eqref{eq:proof2eq2} we obtain
\begin{equation*}
\begin{split}
\cerd^w_U(v_0)&\leq  \cer^w_{\powut{\gamma(1+\varepsilon)}}(v_0)+|\gamma|(1+\varepsilon)r^*\\
&\leq \cerd_{\powut{\gamma(1-\alpha)(1+\varepsilon)^2}}(v_0)+|\gamma| \alpha r^* + \varepsilon |\gamma| r^*(1-(2+\varepsilon)(1-\alpha)).
\end{split}
\end{equation*}
Taking $\varepsilon\to 0$ and invoking Lemma \ref{lem:cnv_cer} yields ``$\leq$" inequalities in the desired equality. The reverse inequalities are obtained in an analogous manner but exploiting the lower bound in Lemma \ref{lem:Uasel}. 
Finally, the above also shows we can replace $w$ by $w_\alpha$.\hfill$\square$
\end{proof}
\begin{remark}
Similarly to Theorem \ref{thm:main} and Corollary \ref{coro:initial_wealth}, $\cerd^w_U(v_0)$ above does not depend on $v_0$ and the optimal strategy for the unconstrained problem scales linearly in the initial wealth.

Also, similarly to Theorem \ref{thm:main}, it follows from Lemma \ref{lem:cnv_cer} that when $\gamma<0$ the equality $\cerd^w_{U}=\cerd_{\powut{\gamma(1-\alpha)}}+|\gamma|\alpha r^*$ holds without assuming that $\cerd_{\powut{\gamma(1-\alpha)(1+\delta)}}$ is finite.
\end{remark}
\begin{corollary}
In the setup of Theorem \ref{thm:main_dollars} we have
$$\cerd_{\powut{\gamma(1-\alpha)}}(v_0)=\cerd_{U\circ F_w}(v_0).$$
\end{corollary}
This Corollary follows from the proof of Theorem \ref{thm:main_dollars}.
Naturally, similar statements can be made relating in general $\cerd$ for power utility and for $U$ which satisfies $(i)$ in Assumption \ref{ass:UwS0}. This is not surprising in the light of the results on the so-called turnpike theorems.
In this stream of literature authors study the convergence of the value function and the optimal strategy for the Merton problem of maximising utility of terminal wealth as the horizon $T$ tends to infinity.
In particular, Hubermann and Ross \cite{HuRo} argue that, in the case of a complete discrete market, the convergence of optimal strategies is equivalent to the convergence of the relative risk aversion, i.e. $-\frac{x U''(x)}{U'(x)} \rightarrow 1-\gamma$, which is essentially $(i)$ in Assumption \ref{ass:UwS0}. Huang and Zariphopoulou \cite{HuZa} study the problem for a continuous time complete market model with deterministic coefficients, as in Section \ref{sec:example_complete}. They find sufficient conditions on $U$ for the optimal strategy to converge to the optimal strategy coming from the problem with a power utility. The analysis in a recent paper of Guasoni and Robertson \cite{GuasoniRobertson:11} includes also incomplete markets. In comparison, our results apply in a much more general context but are also much weaker. Problem \ref{pb:utmax} looks at maximising the long-term asymptotic growth rate of the expected utility and the above Corollary shows that the resulting value function is the same when 
two utility functions have the same asymptotic behaviour. It does not say anything precise about finite horizon utility maximisation and its convergence.

\section{Logarithmic Utility}\label{sec:logut}

So far we have only considered utility functions with constant sign and which dominated a power utility, as in Assumption \ref{ass:U}. In this section, we consider utility functions akin to $U(x)=\log x$. The results are very close in spirit to the ones in the previous two sections, but in fact require less technicalities in the proofs. We go back to the general setup of Section \ref{sec:market_model} and introduce a modified version of maximisation criterion in \eqref{eq:opt_prob}.
\begin{prob}\label{pb:utmax_log}
Given $v_0>0$, a drawdown function $w$ and function  $U$ compute
\begin{equation}\label{eq:opt_prob_log}
\begin{split}
\widetilde{\cer}_{U}(v_0) &:= \sup_{V\in \A(v_0)}\widetilde{\R}_{U}(V),\\
\widetilde{\cer}^w_{U}(v_0) &:= \sup_{V\in \A^w(v_0)}\widetilde{\R}_{U}(V),\\
&\quad \textrm{where}\quad \widetilde{\R}_{U}(V):=\limsup_{T\to \infty} \frac{1}{T} \e\left[U\left(V_T\right)\right],
\end{split}
\end{equation}
along with the optimal wealth processes which achieve the supremum.
\end{prob}
\begin{theorem}\label{thm:main_log} Let $v_0>0$, $w$ be a drawdown function and $U$ a utility function satisfying
\begin{equation*}
\limsup_{x\to\infty} xU'_+(x)<\infty \quad \textrm{and} \quad \liminf_{x\to \infty} \frac{U(x)}{\log(x)} > 0.
\end{equation*}
Let $F_w$ be as in Theorem \ref{thm:main} then
\begin{equation*}
\widetilde{\cer}^w_{U}(v_0)=\widetilde{\cer}_{U\circ F_w}(v_0)
\end{equation*}
and if $V^*$ achieves the maximum in the unconstrained problem then $M^{F_w}(V^*)$ achieves the maximum in the $w$-drawdown constrained problem.
\end{theorem}
\begin{remark}
The equality between value functions in particular states that they are either both finite or both infinite.
\end{remark}

\begin{proof}
We write $K=K_w$ and $F=F_w$. The first part of the proof is identical to the first part of the proof of Theorem \ref{thm:main} and yields
\begin{equation*}
 \widetilde{\cer}^w_{U}(v_0) \geq \widetilde{\cer}_{U \circ F}(v_0).
\end{equation*}
Also, if we had the desired equality then $M^{F}(V^*)$ is optimal for constrained problem when $V^*$ is optimal for the unconstrained one.

It follows from the assumptions that for any $y>0$ there exists $\gamma$ such that
$$xU'_{+}(x)<\gamma<\infty,\quad x \geq y.$$
Applying Lemma \ref{lemma:ae_u} to $e^{U(x)}$,
we deduce that 
\begin{equation}
U(\lambda x) \leq U(x) + \gamma \log \lambda, \text{  } \lambda > 1, \ x \geq y. \label{eq:U_log} 
\end{equation}
Define $w_n, K_n, F_n$ as in the proof of Theorem \ref{thm:main} and recall the computation in \eqref{eq:FAEbound}. It follows that 
$x{U\circ F_n}'_+(x)<\gamma$ for $x\geq v_0$. The same reasoning holds with $F$ in place of $F_n$. We deduce that \eqref{eq:U_log} holds with $U\circ F$ in place of $U$ and $\gamma'$ instead of $\gamma$ and likewise for $U\circ F_n$ and $\gamma_n$.

Let $X\in \A^w(v_0)$, $Y^n=M^{K_n}(X)$ and recall from the proof of Theorem \ref{thm:main} that
\begin{equation*}
Y^n_{t} \geq \frac{1}{1+n}K_n \left(\overline{X}_t\right). \label{eq:y_n>x_log}
\end{equation*}
Using \eqref{eq:U_log} for $U\circ F_n$ and $y=v_0/(n+1)$, $\lambda=n+1$, noting that $K_n(\ovl{X}_t) \geq v_0$, we obtain
\begin{equation*}
\begin{split}
U \left(\ovl{X}_t \right)&= U\circ F_n\circ K_n \left(\ovl{X}_t \right) \leq U \circ F_n\left(\frac{1}{1+n} K_n(\ovl{X}_t )\right) +\gamma_n\log(n+1)\\ &\leq U \circ F_n \left( Y^n_{t} \right) +\gamma_n\log(n+1).
\end{split}
\end{equation*}
Taking expectations, dividing by $t$ and passing to the limit $t\to\infty$ yields:
\begin{equation*}
\widetilde{\R}_U(X)\leq \widetilde{\R}_U(\ovl{X}) \leq \widetilde{\R}_{U\circ F_n}(Y^n).
\end{equation*}
Taking supremum over $X\in \A^w(v_0)$ we conclude that $$\widetilde{\cer}^w_{U}(v_0)\leq \widetilde{\cer}_{U\circ F_n}(v_0).$$

It remains to establish the convergence in $n$ on the LHS. An analogous argument to Lemma \ref{lem:2U} shows that 
to compute $\widetilde{\cer}_{U\circ F_n}(v_0)$ it suffices to consider $V\in \A(v_0)$ such that $V\geq v_0/2$. The bound in \eqref{eq:FFnbound} gives
$$U(c_0F(v))\leq U\circ F_n(v)\leq U\circ F\left(\frac{1}{\epsilon} v^{1+1/n}\right),\quad v\geq \epsilon $$
for arbitrary $\epsilon\in (0,1)$ and some $0<c_0<1$. Taking $\epsilon<v_0/2$ and using \eqref{eq:U_log} for $U$ and $U\circ F$ with $y=\min\{\epsilon,c_0F(\epsilon)\}$ we obtain
\begin{equation}\label{eq:UFlogbound}
U\circ F(v)+\gamma\log c_0\leq U\circ F_n(v)\leq U\circ F(v)+\gamma'\log(v^{1/n}/\epsilon),\quad v\geq \epsilon.
\end{equation}
Finally, consider $\log F(x)/\log x$ for large $x$ and let $z=F(x)$. Then, using \eqref{eq:w-assumption},
$$\frac{\log F(x)}{\log x}=\frac{\log z}{\log K(z)}=\frac{\log z}{\log v_0+\int_{v_0}^z \frac{du}{u-w(u)}}\geq \frac{\log z}{\log v_0+\frac{1}{1-\alpha_1}\log z/v_0},$$
which can be made arbitrary close to $1-\alpha_1>0$ by considering $z$ large enough. Using the assumption on $U$ we conclude that
$$\liminf_{x\to \infty}\frac{U\circ F(x)}{\log x}=\liminf_{x\to \infty}\frac{U\circ F(x)}{\log F(x)}\frac{\log F(x)}{\log x}>0.$$
It follows that for some positive constants $c,c_1$ we have $cU\circ F(x)+c_1\geq \log(x)$ for all $x\geq \epsilon$. 
Combined with \eqref{eq:UFlogbound}, this shows that 
$$\widetilde{\cer}_{U \circ F}(v_0)\leq \widetilde{\cer}_{U \circ F_n}(v_0)\leq (1+c\gamma'/n)\widetilde{\cer}_{U \circ F}(v_0).$$
Taking $n\to \infty$ establishes the desired convergence.\hfill$\square$
\end{proof}
We close this section with a result similar to Theorem \ref{thm:main_dollars}. The definitions of $\widetilde{\cerd}_U$, $\widetilde{\cerd}_U^w$ should be clear and the proof follows closely the arguments in Section \ref{sec:main_dollars} and we omit it for the sake of brevity.
\begin{theorem} \label{thm:log_dollars} Let $U$ be a utility function with $xU'_+(x) \rightarrow \gamma \in (0,\infty)$ as $x \rightarrow \infty$, $v_0>0$, $w$ a drawdown function, $F_w$ as in Theorem \ref{thm:main} and assume $(ii)$ and $(iii)$ in Assumption \ref{ass:UwS0} hold.
 Then
\begin{equation*}
\widetilde{\cerd}^w_{U}(v_0)= \gamma(1-\alpha) \widetilde{\cerd}_{\log}(v_0) + \gamma \alpha r^*
\end{equation*}
and $M^{F_w}(V^*)\in \A^w(v_0)$ achieves the maximum in the drawdown constrained problem if $V^*\in\A(v_0)$ achieves the maximum in the unconstrained problem.
\end{theorem}

\section{Examples}\label{sec:examples}
We discuss now some examples. Our aim is twofold. First, we want to give an example of a rather general setup in which sufficient conditions can be found which guarantee finiteness of $\cer$ for the unconstrained problem, as assumed in Theorem \ref{thm:main}. Second, we want to discuss specific examples when the unconstrained, and hence also the drawdown constrained, portfolio optimisation problem is solved explicitly. In particular we relate our results and methods to the ones in \cite{CK}.
\subsection{Market with price deflators}
We start by assuming existence of a price deflator (or a state price density) process. In the setup of Section \ref{sec:market_model} we further assume that all $S^i_t$ are continuous and that there exists a $\p$-local martingale $(Z_t)$, $Z_t>0$ for all $t\geq 0$, such that $(Z_tS_t^i)$ are $\p$-local martingales, $i=1,\ldots, d$. Note that we do not necessarily assume that $(Z_t)$ is a true martingale and hence that an equivalent martingale measure exists. Our setup is in fact analogous to the most general setup in which stochastic portfolio optimisation is considered, see Fernholz and Karatzas \cite{FerKar}.
Note that if $(V_t)\in \A(v_0)$ then
$$\td (Z_t V_t)= V_t\td Z_t + Z_t \pi_t \td S_t + \pi_t \td \langle Z, S\rangle_t = (V-\pi_t S_t)\td Z_t + \pi_t \td (Z_t S_t),$$
so that $(Z_t V_t)$ is a positive $\p$-local martingale and hence a supermartingale.
Karatzas and Kardaras \cite{KarKar} show that the existence of $(Z_t)$ is equivalent to the NUPBR condition (No Unbounded Profit With Bounded Risk). This condition is weaker than the usual NFLVR condition from Delbaen and Schachermayer \cite{DelSch} and allows for some (very mild) arbitrage opportunities, see examples constructed in \cite{KarKar}. Recall that $\powut{p}(x)=\frac{1}{p}x^p$.
\begin{lemma}\label{lem:suff_finite_cer}
The following implications hold for any $p<1$, $p\neq 0$, and $v_0>0$
\begin{equation*}
\begin{split}
\R_{\powut{-p/(1-p)}}(Z)<\infty \quad &\Longrightarrow \quad \cer_{\powut{p}}(v_0)<\infty,\\
\R_{\powut{-p/(1-p)}}(Z/N)<\infty \quad &\Longrightarrow \quad \cerd_{\powut{p}}(v_0)<\infty.
\end{split}
\end{equation*}
\end{lemma}
\begin{proof}
Let $(V_t)\in \A(v_0)$ so that $(Z_t V_t)$ is a $\p$-local martingale, as above. For $p<0$ we have
\begin{equation*}
\begin{split}
\e [V_T^p]&=\e[(Z_T)^{-p}(Z_TV_T)^p]\\
& \geq \left(\e[(Z_T)^{-\frac{p}{1-p}}]\right)^{(1-p)}\left(\e[Z_T V_T]\right)^{p}\geq v_0^p \left(\e[(Z_T)^{-\frac{p}{1-p}}]\right)^{(1-p)},
\end{split}
\end{equation*}
where we used H\"older's inequality, whose direction is reversed for $p<0$, and the fact that a non-negative local martingale is a supermartingale. The inequalities  above are reversed when we divide both sides by $\frac{1}{p}<0$ and the claim follows. The case $p\in (0,1)$ is even more straightforward --- it is enough to use H\"older's inequality which directly gives the desired inequality (reverse from the one displayed above). The case with numeraire is entirely analogous.\hfill $\square$
\end{proof}
The lemma above gives an example of sufficient conditions to apply Theorems \ref{thm:main} and \ref{thm:main_dollars} since they both require that $\cer_{G}$ or $\cerd_{G}$ is finite. For the latter $G$ is a power utility function and Lemma \ref{lem:suff_finite_cer} applies directly. For the former we would need to bound $G$ by a power utility.

Naturally, in the current very general setup there might be little hope to compute $\cer_{U}$ or find the optimal wealth process. However, one might expect this to be the simplest portfolio optimisation problem to solve. The strength of our results is to show that solving the seemingly much more complex problem with drawdown constraint on wealth paths is in fact equally simple (or hard).

Karatzas and Kardaras \cite[Theorem 4.12]{KarKar} also show that the existence of $(Z_t)$ is equivalent to the existence of a benchmark numeraire $\tilde{N}$ such that $V/\tilde N$ is a supermartingale for any $V\in \A(v_0)$, see also Christensen and Larsen \cite{ChristensenLarsen:07}. This readily implies that  $\widetilde{\cer}_{\log}(v_0)=\widetilde{\R}_{\log}(\tilde N)$ and $\widetilde{\cerd}_{\log}(v_0)=\widetilde{\R}_{\log}(\tilde N N)$.
Indeed, considering $V\in \A(v_0)$ and applying Jensen's inequality gives
\begin{equation*}
\begin{split}
 \limsup_{T \rightarrow \infty} \frac{1}{T} \e \log V_T &\leq \limsup_{T \rightarrow \infty} \frac{1}{T} \e \log  \tilde N_T+ \limsup_{T \rightarrow \infty} \frac{1}{T} \e \log \frac{V_T}{\tilde N_T}\\& \leq
 \limsup_{T \rightarrow \infty} \frac{1}{T} \e \log  \tilde N_T.
\end{split} \end{equation*}
 This observation essentially goes back to Bansal and Lehmann \cite{BansalLehmann:97}.
In a no-arbitrage complete market model, taking $N$ to be the savings account, $(Z_t)$ is the density $\frac{\td \mathbb{Q}}{\td \p}$ where $\mathbb{Q}$ is the equivalent martingale (risk-neutral) measure. Completeness means that $(Z_t)^{-1}$ is an admissible wealth process and thus the benchmark numeraire. In particular, in the setting of Theorem \ref{thm:log_dollars}, we then have
\begin{equation*}
 \widetilde{\cerd}^w_U (v_0) = \gamma r^*-\gamma (1-\alpha) \limsup_{T \rightarrow \infty} \frac{1}{T} \e \log Z_T .
\end{equation*}
It may be natural to start modelling by simply requiring that the benchmark numeraire $\tilde N$ exists. This is pursued in the so-called \emph{benchmark approach}, see Platen and Heath \cite{PlatenHeath:06}. It is clear that Lemma \ref{lem:suff_finite_cer} remains true in this approach when we replace $DZ$ by $1/\tilde N$.

\subsection{Complete market model with deterministic coefficients}\label{sec:example_complete}
We consider now the classical complete financial market model with deterministic coefficients. Let $W_t=(W^1_t,\ldots,W^d_t)\tr$ be a standard $d$-dimensional Brownian motion and $(\F_t)$ the augmentation of its natural filtration. Here $\tr$ denotes vector transpose. $N_t=\exp(\int_0^t r_u\td u)$ is deterministic and $\frac{1}{T}\int_0^T r_u \td u \to r^*\geq 0$ as $T\to \infty$. Each asset follows dynamics given by
$$\frac{\td \tilde S^i_t}{\tilde S^i_t}= \mu^i_t \td t + \sum_{j=1}^d \sigma^{ij}_t\td W^j_t,\quad \tilde S^i_0=s^i_0>0$$
where $\mu^i_t$ and $\sigma^{ij}_t$ are bounded deterministic functions and $\sigma_t$ is invertible. Recall Definition \ref{def:wealth} of wealth process and let $\tilde \pi_t^i:= \pi^i_t S^i_t/V_t$ be the proportion of wealth invested in the $i^{\textrm{th}}$ asset so that $\td V_t= \sum_{i=1}^d \tilde\pi^i_t V_t \frac{\td S^i_t}{S_t^i}$.
The market price of risk is given as $\theta_t:= \sigma^{-1}(\mu_t-r_t\mathbb{I})$, where $\mathbb{I}$ is a $d$-dimensional vector with all entries equal to one. We assume $\theta_t$ is also bounded and that
$$||\theta^*||^2:= \lim_{T\to\infty}\frac{1}{T} \int_0^T ||\theta_u||^2\td u\quad \textrm{ is well defined and finite.}$$
The state price density
$$Z_t:= \exp\left\{-\int_0^t \theta_u\tr \td W_u - \frac{1}{2}\int_0^t ||\theta_u||^2 \td u\right\}$$
is a $\p$--martingale which defines for every $T\in (0,\infty)$ a unique risk neutral measure up to time $T$ via $\frac{\td \Q}{\td \p}|_{\F_T}=Z_T$. To solve $\cer_{\powut{p}}$ one first considers the problem of maximising the expected utility of wealth at a given horizon $T$. The solution is obtained using, by now standard, convex duality arguments, see Karatzas, Lehoczky and Shreve \cite{KaratzasLehochkyShreve:87} or Karatzas and Shreve \cite[Sec.~3.5--3.8]{KaSh}. The optimal wealth process $V^*$ is described explicitly via
\begin{equation}
\tilde\pi^*_t=\frac{1}{1-p}\theta_t\tr\sigma^{-1}_t \label{eq:pi_compl}
\end{equation}
and in particular it is independent of the time horizon $T$. We conclude that it is also optimal for the long-term asymptotic growth rate optimisation. Taking limit of the value functions for the finite horizon problem we obtain
$$\cerd_{\powut{p}}(v_0)=\R_{\powut{p}}(NV^*)=|p|r^*+ \frac{|p|}{2(1-p)}||\theta^*||^2.$$
Note that the difference of a factor $|p|$ when compared to \cite{GZ,CK} is immediate since they consider $ \frac{1}{|p|} \R_{\powut{p}}(NV)$ instead of $\R_{\powut{p}}(NV)$.

Applying Theorem \ref{thm:main_dollars} for a utility function $U$ and a drawdown function $w$, which satisfy Assumption \ref{ass:UwS0}, we obtain
$$\cerd_{U}^w(v_0)=\cerd_{\powut{\gamma(1-\alpha)}}(v_0)+|\gamma|\alpha r^*=
|\gamma|\left(r^*+\frac{(1-\alpha)}{2(1-\gamma(1-\alpha))}||\theta^*||^2\right)$$
which is achieved by the optimal wealth process $X=M^{F_w}(V^*)$. Using \eqref{eq:SDEwDD} we see that
$$\td X_t= \left(X_t - w(\ovl{X}_t)\right)\sum_{i=1}^d \left(\frac{1}{1-\gamma(1-\alpha)}\theta_t\tr\sigma^{-1}_t\right)^i \frac{\td S^i_t}{S^i_t}.$$
In particular, we recover Theorem 5.1 in \cite{CK} by taking $U=\powut{\gamma}$, $\gamma\in (0,1)$ and $w(x)=\alpha x$. It is insightful to understand better how the objects in \cite{CK} relate to the tools of our paper. In fact the Auxiliary problem introduced and solved in \cite{CK} is nothing else but $\cerd_{U\circ F}(v_0)=\cerd_{\powut{\gamma(1-\alpha)}}(v_0)$.
Indeed, the process $N^\pi_\alpha$ defined in $(4.1)$ therein is simply $N M^K(X) $ and $\hat{\pi}_t = \frac{1}{1 - \gamma(1-\alpha)} \theta'_t \sigma^{-1}_t$.

\subsection{Incomplete market example}\label{sec:example_incomplete}
In a recent paper Guasoni and Robertson \cite{GuasoniRobertson:11} solve the unconstrained portfolio optimisation problem for an investor with a power utility in a rather general diffusion model. Our results allow to solve $w$--drawdown constrained problem in their setting. The solution in \cite{GuasoniRobertson:11} is involved and we do not cite the details here for the sake of brevity. Instead we propose to study an application of Theorem \ref{thm:main_dollars} in an incomplete market example adapted from the risk-sensitive control approach in Fleming and Sheu \cite{FleSheu}.

Consider a market with constant interest rate $r$ and one risky asset $\tilde S(t)$ evolving according to
\begin{eqnarray*}
 \frac{\td \tilde S(t)}{\tilde S(t)} &=& (\mu_1 + \mu_2 x(t))\td t + \sigma \td W^1_t + \rho \td W^2_t, \\
 \td x(t) &=& b x(t) \td t + \td W^1_t,
\end{eqnarray*}
where $W^1,W^2$ are two independent Brownian motions and $x(t)$ has an interpretation of an economical factor.

In Theorem 3.1 in \cite{FleSheu} the authors provide a link between Problem \ref{pb:utmax} with a power utility function $\powut{\gamma}$ and a viscosity solution of the dynamic programming equation. In Theorem 4.1 the optimal investment policy is found. We refer the reader to \cite{FleSheu} for further details of the method. In the last section of their paper Fleming and Sheu consider Vasicek interest rate model with a single stock and give an explicit solution to the utility maximisation problem. Our model above is slightly different but we are still able to use their solution.

The difference with Fleming and Sheu \cite{FleSheu} example is that the interest rate is given by $r(t) = r$ in our work and by $r(t) = \lambda x(t) - \frac{b_1}{b_2}$ in theirs, which  requires us to change some coefficients in final formulae in \cite{FleSheu}.
More precisely, assume $\gamma < 0$ and $\mu^2_2 \geq \sigma^2 (K^{(\gamma)})^2$ where $K^{(\gamma)}$ is defined below. Then the value function is equal to
\begin{equation*}
\cerd_{\powut{\gamma}}(v_0) = \frac{1}{2} K^{(\gamma)} + \frac{1}{2}|\eta|^2 + 1/2 \frac{\gamma}{1-\gamma}\frac{(\mu_1-r + \sigma \eta)^2}{\sigma^2 + \rho^2} + |\gamma| r,
\end{equation*}
where
\begin{equation*}
\eta = - \frac{\gamma}{1-\gamma} \frac{\mu_2 + K^{(\gamma)} \sigma (\mu_1-r)}{(D^{(\gamma)} + K^{(\gamma)} E^{(\gamma)})(\sigma^2 + \rho^2)}
\end{equation*}
and
\begin{eqnarray*}
 E^{(\gamma)} &=& 1+ \frac{\gamma}{1-\gamma} \frac{\sigma^2}{\sigma^2 + \rho^2},\\
 K^{(\gamma)} &=& - \frac{b + \frac{\gamma}{1-\gamma} \frac{1}{\sigma^2 + \rho^2} \mu_2 \sigma}{1 + \frac{\gamma}{1-\gamma}\frac{\sigma^2}{\sigma^2 + \rho^2} }
- \frac{1}{\sqrt{1 + \frac{\gamma}{1-\gamma}\frac{\sigma^2}{\sigma^2 + \rho^2}}} \\
&& \cdot \left( -\frac{\gamma}{1-\gamma} \frac{\mu^2_2}{\sigma^2 + \rho^2} + \frac{(b + \frac{\gamma}{1-\gamma} \frac{1}{\sigma^2+ \rho^2})^2}{1+ \frac{\gamma}{1-\gamma}\frac{\sigma^2}{\sigma^2+\rho^2}}\right)^{1/2},\\
 D^{(\gamma)} &=& - \sqrt{1+ \frac{\gamma}{1-\gamma} \frac{\sigma^2}{\sigma^2 + \rho^2}} \left( -\frac{\gamma}{1-\gamma} \frac{\mu^2_2}{\sigma^2+\rho^2} + \frac{(b+ \frac{\gamma}{1-\gamma}\frac{\mu_2 \sigma}{\sigma^2+\rho^2})^2}{1+\frac{\gamma}{1-\gamma}\frac{\sigma^2}{\sigma^2+\rho^2}}\right)^{1/2}.
\end{eqnarray*}
And the optimal investment policy $\tilde \pi_t$, which is fraction of wealth invested in risky asset at time $t$, is given by
\begin{equation*}
 \tilde \pi_t = D^{(\gamma)} x(t) + a^{(\gamma)},
\end{equation*}
for some constant $a^{(\gamma)}$.

In the setting of Theorem \ref{thm:main_dollars} we obtain that
\begin{equation*}
 \cerd^w_U (v_0) = \frac{1}{2} K^{(\gamma (1-\alpha))} + \frac{1}{2}|\eta|^2 + 1/2 \frac{\gamma (1-\alpha)}{1-\gamma (1-\alpha)}\frac{(\mu_1 -r + \sigma \eta)^2}{\sigma^2 + \rho^2} + |\gamma| (1-\alpha)r,
\end{equation*}
where $\gamma<0$ and $\alpha\in (0,1)$ are defined in Assumption \ref{ass:UwS0}. Note that we could also consider $\gamma\in (0,1)$ under the additional parameter restriction which makes appropriate $K^{(\cdot)}$ well defined.

\appendix
\section{Appendix}\label{app}
\normalsize

We state and prove here lemmas used in the proofs in the main body of the paper. Note however that the first two lemmas may be of independent interests. 
Lemma \ref{lem:2U} shows that computing $\cer_U$ it is enough to consider wealth processes which dominate a given fraction of the numeraire. Lemma \ref{lem:cnv_cer} studies convergence of $\cer_{U_n}\to\cer_U$ as $U_n\to U$. 

Recall the general setup introduced in Section \ref{sec:market_model} and objectives $\cer$ and $\cerd$ given in \eqref{eq:opt_prob_DD}--\eqref{eq:opt_prob} and \eqref{eq:utmax_dollars} respectively.

\begin{lemma}\label{lem:2U}
Let $U$ be a continuous non-decreasing function with a well defined locally bounded right derivative $U'_+$. Assume $U$ is either positive or negative and satisfies
\begin{equation}\label{eq:AEU_appendix}
\limsup_{x\to \infty} \frac{xU'_+(x)}{|U(x)|} < \infty.
\end{equation} 
Then for any $0<y<v_0$ 
\begin{eqnarray*}
 \sup_{V \in \A(v_0)} \R_U (V) &=&  \sup_{V \in \A(v_0),\ V \geq y} \R_U (V), \\
\cer_{U}(1) &=& \cer_{U}(v_0).
\end{eqnarray*}
Further, if $N$ is non-decreasing, or in general if $N_t\geq \underline{N}$ for all $t\geq 0$ and some constant $\underline{N}>0$ then also
\begin{eqnarray*}
 \sup_{V \in \A(v_0)} \R_U (NV) &=&  \sup_{V \in \A(v_0),\ NV \geq y} \R_U (NV),\\
\cerd_{U}(1) &=& \cerd_{U}(v_0).
\end{eqnarray*}
\end{lemma}

\begin{proof}

For $V \in \A(v_0)$ and some $0<\varepsilon <1$ consider the process $\tilde V \in \A(v_0)$ given by\footnote{Recall from Example \ref{ex:const_dd} that $\tilde V$ is also an Az\'ema-Yor transformation $\tilde V= M^F(V)$ corresponding to an affine $F$.} $\tilde{V}_t = \varepsilon v_0 + (1-\varepsilon) V_t \geq \varepsilon v_0$, $t\geq 0$. As $U$ satisfies \eqref{eq:AEU_appendix} and $\tilde{V}_t \geq \varepsilon v_0$ we are able to use Lemma \ref{lemma:ae_u} to deduce that for some non-zero $\gamma \in \mathbb{R}$
\begin{equation*}
(1-\varepsilon)^{\gamma}U\left(V_t\right)\leq (1-\varepsilon)^{\gamma}U\left(\frac{1}{1-\varepsilon}\tilde{V}_t\right) \leq U\left(\tilde{V}_t\right),
\end{equation*}
where we used $\frac{1}{1-\varepsilon}\tilde V \geq V$.
Taking expectation, applying $\frac{1}{T} \log$ and taking limit as $t\to\infty$, we deduce that
\begin{equation} \label{eq:Vtil_V}
\R_U(V)\leq \R_U(\tilde{V}).
\end{equation}
Thus, taking $\varepsilon = y/v_0$ we obtain
\begin{equation*}
 \sup_{V \in \A(v_0)} \R_U (V) \leq  \sup_{V \in \A(v_0), \text{ s.t.} V \geq y} \R_U (V)
\end{equation*}
and the reverse inequality is trivial and the first equality in the statement follows.

Now consider $V \in \A(v_0)$ such that $V \geq y$ for some $y>0$. Note that $V^1:=\frac{1}{v_0}V\in \A(1)$. 
If $v_0 >1$, Lemma \ref{lemma:ae_u} and monotonicity of $U$ yield, for some $\gamma \in \re\setminus\{0\}$,
\begin{equation*}
U(V^1_t)\leq U(V_t) = U(v_0 V^1_t) \leq v_0^{\gamma} U(V^1_t),\quad t\geq 0.
\end{equation*}
If $0<v_0<1$ we obtain similarly 
\begin{equation*}
v_0^\gamma U(V^1_t)=v_0^\gamma U\left(\frac{1}{v_0}V_t\right)\leq U(V_t)\leq U(V^1_t).
\end{equation*} 
It follows that
\begin{equation*}
 \sup_{V \in \A(v_0),\ V \geq y} \R_U (V) =  \sup_{V^1 \in \A(1),\ V^1 \geq y/v_0} \R_U (V^1).
\end{equation*}
The equality $\cer_{U}(1) = \cer_{U}(v_0)$ now follows from the first equality in the statement which we established above.

For the second pair of equations in the statement of the Lemma, we have $\tilde{V}_tN_t = \varepsilon v_0 N_t + (1-\varepsilon) V_tN_t \geq \varepsilon v_0 N_t \geq \underline{N}\varepsilon v_0>0$ and the arguments are then entirely analogous. 
\hfill $\square$
\end{proof}

\begin{lemma}\label{lem:cnv_cer}
Let $U_n,U$ be nondecreasing functions of the same sign, continuous with a well defined locally bounded right derivative, satisfying Assumption \ref{ass:U} and \eqref{eq:AEU_appendix}. Assume further that for some $c,c_1>0$ and some $0<\nu<1$
\begin{equation*}
\forall\ \delta>0\ \exists n_\delta\ \forall n\geq n_\delta\quad c_1 U(x)\leq U_n(x)\leq U(cx^{1+\delta}),\ x\geq \nu.
\end{equation*}
If $U_n,U$ are negative we have, for any $v_0>0$,
\begin{equation}\label{eq:cnv_cer}
\begin{split}
\cer_{U_n}(1)&\van  \cer_{U}(1)
\end{split}
\end{equation}
If $U_n,U$ are positive the above holds assuming that $\cer_{G}(1)<\infty$, where $G(x): = U (x)^{1+\delta}$ for some $\delta>0$. Consequently, we then have $\cer_{U}(1)<\infty$.

If $N$ is bounded away from zero, $N_t\geq \underline N$, $t\geq 0$ for some $\underline{N}>0$, then the above results hold with $\cer$ replaced by $\cerd$. 
\end{lemma}
\begin{proof}

We prove the statement for $\cer$ and $\cerd$ simultaneously. They follow respectively by taking $\xi = V$ and $\xi=VN$, $V\in \A(1)$, in what follows. When considering the latter the assumption that $N$ is bounded away from zero is in place. 
Observe that, by Lemma \ref{lem:2U} it is sufficient to consider $\xi\geq \nu$.

Assume that for $n$ and $K$ large enough and any $\xi \geq \nu$ we have
\begin{equation}
\begin{split}
\label{eq:R_cutoff}
\R_{ U_n}(\xi)&= \limsup_{T \rightarrow \infty} \frac{1}{T} \log \e [U_n\left(\xi_T\right)]
= \limsup_{T \rightarrow \infty} \frac{1}{T} \log \e [ U_n\left(\xi_T\right) \mathbf{1}_{\xi_T \leq K^T}], \\
\R_{ U}(\xi)&= \limsup_{T \rightarrow \infty} \frac{1}{T} \log \e [U\left(\xi_T\right)]
= \limsup_{T \rightarrow \infty} \frac{1}{T} \log \e [ U\left(\xi_T\right) \mathbf{1}_{\xi_T \leq K^T}].
\end{split}
\end{equation}
Take $\delta>0$, large $K,T,n$ so that the assumptions yield
\begin{equation*}
\begin{split}
  c_1U(\xi_T)\mathbf{1}_{\xi_T \leq K^T}\leq  U_n(\xi_T)\mathbf{1}_{\xi_T \leq K^T}& \leq U(c\xi_T^{1+\delta})\mathbf{1}_{\xi_T \leq K^T}\leq U(cK^{\delta T}\xi_T)\mathbf{1}_{\xi_T \leq K^T}\\
   &\leq (cK^{\delta T})^{\gamma} U(\xi_T)\mathbf{1}_{\xi_T \leq K^T},
\end{split}
\end{equation*}
where we used Lemma \ref{lemma:ae_u} to obtain the last inequality. Recall that we defined $\log x = - \log(-x)$ for $x<0$. Taking expectation, applying $\frac{1}{T}\log$ and taking the limit as $T\to\infty$ in the above we conclude, thanks to \eqref{eq:R_cutoff}, that 
$$\R_U(\xi)\leq \R_{U_n}(\xi)\leq \R_U(\xi)+|\gamma|\delta \log K.$$
This is true for $n$ large enough and any $\xi$ and hence also when we take supremum over $\xi$. We deduce \eqref{eq:cnv_cer} taking $\delta \to 0$.

It remains to argue \eqref{eq:R_cutoff}. We will prove this separately for positive and negative $U$. Consider first $U_n,U\geq 0$. Assumption \ref{ass:U} implies that there exist $\tilde{c} >0$ and $\varepsilon > 0$ such that $\nu \leq x < \tilde{c}U(x)^{1/\varepsilon}$. For any $\delta'>0$, using Lemma \ref{lemma:ae_u}, we obtain
\begin{eqnarray*}
U(x^{1+\delta'})^{1+\delta'}  \leq & U\left(\left(\frac{x}{\nu}\right)^{\delta'} x\right)^{1+\delta'}\leq \left(\left(\frac{x}{\nu}\right)^{\gamma\delta'}U(x)\right)^{1+\delta'}\\
  \leq & \left(\frac{\tilde{c}}{\nu}\right)^{\gamma \delta' (1+\delta')}U(x)^{1+\delta'(1+\frac{\gamma'(1+\delta')}{\varepsilon})}, \quad x \geq \nu.
\end{eqnarray*}
From the proof of Lemma \ref{lemma:ae_u} it is clear that $U$ and $\gamma$ have the same sign and we conclude that for some $\delta'\leq \delta$ we have 
\begin{equation} \label{eqn:U_pow_delta}
U(x^{1+\delta'})^{1+\delta'} \leq c_2 U(x^{1+\delta})=c_2 G(x), \quad x \geq \nu,
\end{equation}
for some $c_2 > 0$. \\
Using Chebyshev's inequality we obtain
\begin{equation*}
\p\left(\xi_T \geq K^T\right) \leq \frac{\e U_n(\xi_T)}{U_n(K^T)} \leq \frac{\e U_n(\xi_T)}{\tilde{c}^{-\varepsilon} c_1 K^{\varepsilon T}}.
\end{equation*}
In the last inequality we used again the fact that $U_n(x)\geq c_1 U(x)$ for $x\geq \nu$ and that $U(x) \geq \tilde{c}^{-\varepsilon} x^{\varepsilon}$. Take $n>n_{\delta'}$ with $\delta'$ as in \eqref{eqn:U_pow_delta}. Combining the above and using twice H\"{o}lder's inequality with $p=1+\delta'$, $1/p+1/q=1$, we obtain
\begin{equation*}
\begin{split}
  \e[U_n(\xi_T) \mathbf{1}_{\xi_T \geq K^T}]  & \leq \left(\e U_n(\xi_T)^p\right)^{\frac{1}{p}}\p(\xi_T\geq K^T)^{\frac{1}{q}} \leq
  \left(\e U_n(\xi_T)^p\right)^{\frac{1}{p}}\left( \frac{\e U_n(\xi_T)}{\tilde{c}^{-\varepsilon} c_1 K^{\varepsilon T}}\right)^{\frac{1}{q}}\\
    & \leq \frac{\left(\e U_n(\xi_T)^p\right)^{\frac{1}{p}\left(1+\frac{1}{q}\right)}}{\tilde{c}^{-\varepsilon/q} c_1^{1/q}K^{\varepsilon T/q}}
    \leq \frac{\left(\e U(c\xi^{\delta'+1}_T)^{1+\delta'}\right)^{\frac{1}{p}\left(1+\frac{1}{q}\right)}}{\tilde{c}^{-\varepsilon/q} c_1^{1/q}K^{\varepsilon T/q}}.
\end{split}
\end{equation*}
Let $C_G$ denote $\cer_{G}(1)$ or $\cerd_G(1)$ depending on whether we consider $\xi=V$ or $\xi=VN$.
Let $\gamma'$ be the constant resulting from Lemma \ref{lemma:ae_u} applied with $x_0=v^{1+\delta'}$. We can then continue the above chain of inequalities
\begin{equation*}
\begin{split}
\frac{\left(\e U(c\xi^{\delta'+1}_T)^{1+\delta'}\right)^{\frac{1}{p}\left(1+\frac{1}{q}\right)}}{\tilde{c}^{-\varepsilon/q} c_1^{1/q}K^{\varepsilon T/q}}& \leq \frac{\left(\max\{c,1\}^{\gamma'}\e U(\xi^{\delta'+1}_T)^{1+\delta'}\right)^{\frac{1}{p}\left(1+\frac{1}{q}\right)}}{\tilde{c}^{-\varepsilon/q} c_1^{1/q}K^{\varepsilon T/q}}\\
&\leq  c_3\frac{\exp((C_G+\kappa)(1/p+1/pq)T)}{K^{\varepsilon T/q}}\\
& = c_3 \exp\left(\left((C_G+\kappa)\left(\frac{1}{p}+\frac{1}{pq}\right)-\frac{\varepsilon}{q}\log K\right)T\right),
\end{split}
\end{equation*}
where to get the second inequality we used \eqref{eqn:U_pow_delta} and the fact that for any $\kappa>0$, for $T$ large enough, we have $\e G(\xi_T)\leq \exp((C_G(1)+\kappa)T)$. $c_3$ is a positive constant which is can be made explicit from the above computation. For $K$ large enough, the above is decreasing exponentially in $T$. Combining the two displays above we conclude that for any $\kappa>0$, $n>n_{\delta'}$, $K$ large enough and all $T$ large enough we have $\e[U_n(\xi_T) \mathbf{1}_{\xi_T \geq K^T}]\leq \kappa$ and hence
$$ \e [U_n(\xi_T)]\geq \e[U_n(\xi_T) \mathbf{1}_{\xi_T \leq K^T}] \geq \e [U_n(\xi_T)]-\kappa\geq \e [U_n(\xi_T)]\left(1-\frac{\kappa}{c_1 U(\nu)}\right),
$$
where we wrote $\kappa=\frac{\kappa}{\e [U_n(\xi_T)]}\e [U_n(\xi_T)]\leq \frac{\kappa}{U_n(\nu)}\e [U_n(\xi_T)]$ and used the assumption $U_n\geq c_1U$. The first equality in \eqref{eq:R_cutoff} now follows by taking expectations, applying $\frac{1}{T}\log$ and letting $T\to \infty$. Analogous, but simplified, arguments to the above yield the second equality in \eqref{eq:R_cutoff}. 

It remains to show \eqref{eq:R_cutoff} when $U_n, U < 0$. We detail the arguments for $U_n$ and the first equality in \eqref{eq:R_cutoff}.
Obviously $0\geq U_n\left(\xi_T\right) \mathbf{1}_{\xi_T \leq K^T}\geq U_n\left(\xi_T\right)$ so  \eqref{eq:R_cutoff} holds if $\zeta_n:=\R_{U_n}(\xi)=\infty$. Assume now that $\zeta_n<\infty$ and note also that $\zeta_n \geq 0$ since $\xi \geq \nu$.  
Using Assumption \ref{ass:U} on $U$ we see that there exists $\varepsilon>0$ such that $0>U(x)\geq -\tilde c x^{-\varepsilon}$, $x\geq \nu$. This yields 
$\e[U_n(\xi_T) \mathbf{1}_{\xi_T \geq K^T}] \geq c_1 U(K^T) \geq - c_1\tilde c K^{-T \varepsilon}$.
It follows that
\begin{equation*}
\begin{split}
\e [U_n(\xi_T)] &\leq \e [U_n(\xi_T) \mathbf{1}_{\xi_T \leq K^T}] \leq \e [U_n(\xi_T)] + c_1\tilde c K^{-T \varepsilon}\\
&= \e [U_n(\xi_T)] \left(1 - \frac{c_1\tilde c K^{-T \varepsilon}}{\e |U_n(\xi_T)|}\right) \leq \e [U_n(\xi_T)] (1 - c_4\mathrm{e}^{-(\varepsilon \ln  K - \zeta_n-\kappa)T}),
\end{split}
\end{equation*}
where $c_4=c_1\tilde c$, we took $\kappa>0$ arbitrary and $T$ large enough. Taking $K>\exp((\zeta_n+\kappa)/\varepsilon)$, applying $\frac{1}{T}\log$ and letting $T\to\infty$ we see that \eqref{eq:R_cutoff} holds.

\mbox{}\hfill $\square$

\end{proof}

The following Lemma is a slight extension of the first part of Lemma 6.3 in Kramkov and Schachermayer \cite{KrSch}. 
\begin{lemma} \label{lemma:ae_u} Let $U: (0,\infty) \rightarrow \mathbb{R}$ be a continuous nondecreasing function, either strictly positive or strictly negative, with a well defined and locally bounded right derivative and which satisfies \eqref{eq:AEU_appendix}. Then for any $x_0 > 0$ there exists $\gamma \in \mathbb{R}\setminus\{0\}$ such that 
\begin{equation*}
U(x) \leq U(\lambda x) \leq \lambda ^ {\gamma} U(x) \quad \text{ for all } \lambda > 1,\ x\geq x_0.
\end{equation*}
\end{lemma}
 
\begin{proof} Let $x_0>0$. 
From \eqref{eq:AEU_appendix}, the fact that $U$ is monotone and of constant sign, and $U'_+$ is locally bounded, there exists non-zero $\gamma \in \mathbb{R}$ such that 
\begin{equation*}
 \frac{xU'_+(x)}{\gamma U(x)} < 1 \quad x \geq x_0,
\end{equation*}
where $\gamma$ has the same sign as $U$.

Fix $x \geq x_0$ and define functions $F(\lambda) := U(\lambda x)$ and $G(\lambda) := \lambda^{\gamma} U(x)$ for $\lambda > 1$.
Then, $F(1) = G(1)$ and $F'_+(1) = xU'_+(x) < \gamma U(x) = G'_+(1)$. Hence, $F(\lambda) < G(\lambda)$ for $\lambda \in (1, 1+\varepsilon)$ for some $\varepsilon > 0$. Assume that $F(\lambda) > G(\lambda)$ for some $\lambda\in (1,\infty)$ then from continuity
of $F$ and $G$ there exists a point $\lambda^* >1$ such that $F(\lambda^*) = G(\lambda^*)$ and $F'_+(\lambda^*) \geq G'_+(\lambda^*)$, but
\begin{equation*}
 F'_+(\lambda^*) = x U'_+(\lambda^* x) < \frac{\gamma}{\lambda^*}U(\lambda^* x) = \frac{\gamma}{\lambda^*} F(\lambda^*) = \frac{\gamma}{\lambda^*} G(\lambda^*) = G'_+(\lambda^*),
\end{equation*}
 which gives us a contradiction.\hfill $\square$
\end{proof}

\begin{lemma}\label{lem:Uasel}
Suppose $U$ is a utility function which satisfies the first condition of Assumption \ref{ass:UwS0}.  Then for any $x_0>0,\varepsilon>0$, there exist $c_-,c_+>0$ such that
$$c_- \powut{\gamma(1-\varepsilon)}(x)\leq U(x)\leq c_+ \powut{\gamma(1+\varepsilon)}(x),\quad x\geq x_0.$$
\end{lemma}
\begin{proof}
Let us consider $U>0$, the case of $U<0$ being entirely analogous. 
The assumption on $U$ means that for any $\epsilon>0$, there exists $y_0>0$ such that
\begin{equation*}
\frac{yU'_+(y)}{U(y)} \in \left(\gamma (1- \varepsilon), \gamma (1+ \varepsilon)\right) \text{ for  } y \in [y_0, \infty).
\end{equation*}
For $x\geq y_0$ we express $U(x)$ as
\begin{equation*}
 U(x)= U(y_0)\exp \left\{ \int_{y_0}^{x} \frac{yU'_+(y)}{U(y)}\frac{\td y}{y}  \right\}, \label{eq:u_rep}
\end{equation*}
which gives establishes the claim for $x\geq y_0$ with $c_\pm=U(y_0)/\powut{-\gamma(1\pm\varepsilon)}(y_0)$.
It follows that the claim holds for $x\geq x_0$ for any $x_0>0$ with
\begin{equation*}
\begin{split}
&c_-:= \min \left\{ \frac{U(x)}{\powut{\gamma(1-\varepsilon)}(x)}, x \in [x_0, y_0]\right\}, \\& c_+:= \max \left\{ \frac{U(x)}{\powut{\gamma(1+\varepsilon}(x))}, x \in [x_0, y_0]\right\}.
\end{split}
\end{equation*}
\hfill $\square$
\end{proof}
\textbf{Acknowledgement.} It is our pleasure to thank anonymous referee and Associated Editor at \emph{Finance and Stochastics} whose comments helped us greatly to improve the paper. Jan Ob\l\'oj is grateful to Nicole El Karoui for the stimulating discussions they had when working on \cite{CEO} and from which the initial ideas for this paper  originated.
\bibliographystyle{acm}
\bibliography{bibliography_nolinks}

\end{document}